\documentclass[lettersize,onecolumn]{IEEEtran}
\usepackage{amsmath,amsfonts,amssymb}
\usepackage{algorithmic}
\usepackage{algorithm}
\usepackage{array}
\usepackage[caption=false,font=normalsize,labelfont=sf,textfont=sf]{subfig}
\usepackage{textcomp}
\usepackage{stfloats}
\usepackage{url}
\usepackage{verbatim}
\usepackage{graphicx}
\usepackage{cite}
\usepackage{bm}
\usepackage{xcolor}
 \usepackage{amsthm}

\newtheorem{theorem}{Theorem}
\newtheorem{corollary}{Corollary}
\newtheorem{lemma}{Lemma}
\newtheorem{remark}{Remark}

\renewcommand{\baselinestretch}{0.95}

\newcommand{\s}{\mathsf{s}}
\newcommand{\asf}{\mathsf{a}}
\newcommand{\bsf}{\mathsf{b}}
\newcommand{\csf}{\mathsf{c}}

\newcommand{\rtwo}{I\hspace{-0.5pt}I}
\newcommand{\II}{(\mathrm{I\hspace{-0.5pt}I})}
\newcommand{\I}{{(\mathrm{I})}}
\allowdisplaybreaks

\begin{document}

\title{Asymmetric Encoding-Decoding Schemes for Lossless Data Compression}

\author{Hirosuke Yamamoto,~\IEEEmembership{Life Fellow,~IEEE, and Ken-ichi Iwata,~\IEEEmembership{Member, ~IEEE}}
        % <-this % stops a space
\thanks{The contents of this paper were presented in part at the IEEE ISIT2024 \cite{YI2024-a} and the ISITA2024 \cite{YI2024-b}. 
This work was supported in part by JSPS KAKENHI Grant Numbers 24K07487 and 24K14818.
}
% <-this % stops a space
%\thanks{Manuscript received October 26, 2023; revised December 8, 2023.}
}

%\markboth{Journal of \LaTeX\ Class Files,~Vol.~1, No.~2, December~2023}%

\IEEEpubid{0000--0000~\copyright~2023 IEEE}

\maketitle

\begin{abstract}
This paper proposes a new lossless data compression coding scheme named an asymmetric encoding-decoding scheme (AEDS), which can be considered as a generalization of tANS (tabled variant of asymmetric numeral systems). In the AEDS, a data sequence $\bm{s}=s_1s_2\cdots s_n$ is encoded  
in backward order $s_t, t=n, \cdots, 2,1$, while $\bm{s}$ is decoded in forward order $s_t, t=1, 2, \cdots, n$ in the same way as the tANS. But, the code class of the AEDS is much broader than that of the tANS.
We show for i.i.d.~sources 
that an AEDS with 2 states (resp.~5 states) can attain a shorter average code length than the
Huffman code if a child of the root in the Huffman code tree has a probability weight larger than
0.61803 (resp.~0.56984). Furthermore, we derive several upper bounds on the average code length of the AEDS, which also hold for the tANS, and we show that the average code length of the optimal AEDS and tANS with $N$ states converges to the  source entropy with speed $O(1/N)$ as $N$ increases.
\end{abstract}

\begin{IEEEkeywords}
AEDS, ANS, noiseless data compression.
\end{IEEEkeywords}

\section{Introduction}
The most well-known and commonly used noiseless data compression codes are the Huffman code\cite{Huffman}  and the arithmetic code\cite{Rissanen, Pasco, NiMa}.
However, new codes that are superior to these codes have been proposed recently.
The AIFV (almost instantaneous fixed-to-variable length)  code \cite{YTH2015, WYH2017, SKM, SNYK} can attain a better compression rate 
than the Huffman code by using multiple coding trees and allowing a little decoding delay.
Furthermore, the arithmetic code can be replaced with the ANS (asymmetric numeral systems) 
proposed by Duda \cite{duda2009, duda2014, DTGD, duda2022, PDPCMM, PDPCMM2023}.

In conventional data compression coding such as Huffman coding and arithmetic coding, 
a data sequence $\bm{s}=s_1s_2\cdots s_n$ is encoded and decoded in the same order in the way of $s_t, t=1,2,\cdots, n$.
On the other hand, in the ANS, $\bm{s}$ is encoded in backward order $s_t, t=n, \cdots, 2,1$, while $\bm{s}$ is decoded in forward order $s_t, t=1, 2, \cdots, n$. 
This backward-order encoding allows the ANS to encode and decode data sequences using a single integer variable although an interval of $[0,1)$ (or integers) and its division are used in the encoding and decoding of the arithmetic code. Because of this feature, the ANS can achieve almost the same compression rate as the arithmetic code with less arithmetic operations than the arithmetic code, and the ANS is used in many 
practical systems \cite{PDPCMM2023}. 
Especially, the tANS (tabled variant of ANS) \cite{duda2009, duda2014, DTGD, duda2022, PDPCMM, PDPCMM2023} has the advantage of being able to 
encode and decode data sequences at high speed  by using a lookup table instead of arithmetic operations.
Although there have been little detailed information-theoretical analyses on the ANS, the coding rate has been evaluated information-theoretically for each variant of the ANS in \cite{YI2024} recently.

In this paper, we  propose an asymmetric encoding-decoding scheme (AEDS) that uses backward-order encoding and forward-order decoding, as in the ANS, but does not use arithmetic operations such as Huffman coding. The AEDS can be considered as a generalization of the tANS, but the code class of the AEDS is much broader than that of the tANS. 

In Section II, we introduce the AEDS as a generalization of the tANS.
In Section III, we demonstrate that an AEDS can be easily constructed using an arbitrary code tree,
and an AEDS with 2 states (resp.~5 states) based on the Huffman code tree can attain a shorter average code length than the Huffman code if a child of the root in the Huffman code tree has a probability weight larger than 0.61803 (resp.~0.56984).  We also show that 
when the AEDS is applied to sources with uniform distributions,
an AEDS based on a non-Huffman code tree can attain a shorter average code length 
than the AEDS based on the Huffman code tree.
In Section IV, we consider a restricted AEDS called a state-divided AEDS (sAEDS), and we derive
several upper bounds of the average code length, which also hold for the tANS in addition to the sAEDS.
Finally in Section V, we treat  the case such that the number $N$ of states used in sAEDS is large,
and we show that the average code length of the optimal AEDS and tANS converges to the source entropy with speed $O(1/N)$ as $N$ increases.

All theorems and lemmas shown in Sections III--V are proved in Appendixes \ref{App-A}--\ref{App-H}.
In order to construct  an AEDS, the phased-in code\footnote{
The phased-in code is the Huffman code for a uniformly distributed source over a finite alphabet. But,
the phased-in code is also used for sources with non-uniform or unknown probability distributions. 
Refer \cite{phasedin}\cite{phasedin-2}. The phased-in code is also called the CBT (complete binary tree) code. } 
is often used in this paper. Hence, the average code length of the phased-in code is evaluated
in Appendix \ref{App-I}.

In this paper, $\lg a$ means $\log_2a$, $\lfloor a \rfloor$ (resp.~$\lceil a \rceil$) is the greatest  integer less than (resp.~the least integer larger than) or equal to $a$, and $|\mathcal{A}|$ represents the cardinality of a discrete set $\mathcal{A}$.

\section{AEDS generalized from tANS}
\subsection{tANS}
We first introduce tANS (tabled variant of ANS)\cite{PDPCMM,PDPCMM2023}. 
Let $\mathcal{S}$ be a finite discrete alphabet and $p=\{p(s)\,|\, s\in\mathcal{S}\}$ be a memoryless source
probability distribution. 
In the encoding and decoding of the tANS, the following sets are used.
\begin{itemize}
\item $\mathcal{X}=\{N, N+1, \cdots, 2N-1\}$, where $N$ is a positive integer, is the set of internal states with $|\mathcal{X}|=N$.
\item $\mathcal{X}_s$ is a subset of $\mathcal{X}$ corresponding to $s\in\mathcal{S}$ such that
$\mathcal{X}_s\cap \mathcal{X}_{s'} =\emptyset$ if $s\ne s'$ and $\mathcal{X}= \bigcup_{s\in\mathcal{S}} \mathcal{X}_s$. Letting $N_s=|\mathcal{X}_s|$, we have $N=\sum_{s\in\mathcal{S}} N_s$.
\item $\mathcal{Y}_s=\{N_s, N_s+1, \cdots, 2N_s-1\}$ is another set 
corresponding to $s\in\mathcal{S}$ with $N_s=|\mathcal{Y}_s|$.
\end{itemize}
Usually, $N$ and $N_s$ are selected so that $p(s)\approx N_s/N$ for each $s\in\mathcal{S}$.
Since $\mathcal{X}_s$ and $\mathcal{Y}_s$ have the same size $N_s$, 
we can have one-to-one correspondence between $\mathcal{X}_s$ and $\mathcal{Y}_s$, and hence,
between $(s,y)\in\mathcal{S}\times \mathcal{Y}_s$ and $x\in\mathcal{X}_s\subset \mathcal{X}$.
We represent this correspondence by two functions $x=C[s,y]$ and $(s,y)=D[x]$, and we assume
that these functions $C$ and $D$ are given\footnote{How to construct $\mathcal{X}_s, s\in\mathcal{S}$ and functions
$C$ and $D$ are given in e.g.,~\cite{PDPCMM, PDPCMM2023,DY2019, YD2019, BHSM}.}.

In the tANS, an i.i.d.~data sequence $\bm{s}=s_1s_2\cdots s_n$, $s_t\in\mathcal{S}$ is
encoded in backward order of $s_t$, $t=n, n-1, \cdots, 2, 1$ and is decoded in forward order of 
 $s_t$, $t=1, 2, \cdots, n$ in the following way when $N=2^k$ for an integer $k$.

\vspace{0.5cm}
\noindent{\bf Encoding procedure of tANS}
\begin{enumerate}
\item[1.] Set a data sequence $\bm{s}=s_1s_2\cdots s_n$, and select $x_n\in\mathcal{X}$ arbitrarily.
\item[2.] Repeat \eqref{eq2A-1}--\eqref{eq2A-4} in backward order, starting from $t=n$ down to 1.
\begin{align}
k_t&=\left\lfloor \lg\frac{x_t}{N_{s_t}}\right\rfloor, \label{eq2A-1}\\
\beta_t &= x_t\bmod 2^{k_t}, \label{eq2A-2}\\
y_{t-1}&= \left\lfloor \frac{x_t}{2^{k_t}} \right\rfloor, \label{eq2A-3}\\
x_{t-1}&= C[s_t, y_{t-1}]\label{eq2A-4}
\end{align}
\item[3.] The codeword sequence of $\bm{s}$ is given by $x_0 \beta_1 \beta_2 \cdots \beta_n$.
\end{enumerate}

\vspace{0.5cm}
\noindent{\bf Decoding procedure of tANS}
\begin{enumerate}
\item[1.] Set a codeword sequence $x_0 \beta_1 \beta_2 \cdots \beta_n$.
\item[2.] Repeat \eqref{eq2A-5}--\eqref{eq2A-7} from $t=1$ to $n$ in forward order.
\begin{align}
(s_t, y_{t-1})&=D[x_{t-1}], \label{eq2A-5}\\
k_t&= \left\lfloor \lg \frac{N}{y_{t-1}}\right\rfloor, \label{eq2A-6}\\
x_t&=2^{k_t}y_{t-1}+\beta_t. \label{eq2A-7}
\end{align}
\item[3.] The decoded data sequence is given by $s_1 s_2 \cdots s_n$.
\end{enumerate}

Note that \eqref{eq2A-1}--\eqref{eq2A-4} give a mapping $(s_t, x_t) \mapsto (\beta_t, x_{t-1})$ in encoding, 
and \eqref{eq2A-5}--\eqref{eq2A-7} give its reverse mapping $(\beta_t, x_{t-1}) \mapsto (s_t, x_t)$ in decoding. 
The codeword $\beta_t$ of a source symbol $s_t$ depends on state $x_t$, and $x_{t-1}$ obtained by \eqref{eq2A-4} satisfies $x_{t-1}\in\mathcal{X}_{s_t} \subset \mathcal{X}$.
In the codeword sequence  $x_0 \beta_1 \beta_2 \cdots \beta_n$, $x_0$ can be represented by a fixed length code with $\lceil \lg N\rceil$ bits.
In decoding, $\beta_t$ can be parsed correctly in the codeword sequence since 
$k_t$ represents the
bit length of $\beta_t$ which can be obtained by \eqref{eq2A-6}. 
Refer, e.g., \cite{PDPCMM, PDPCMM2023, YI2024} for more details of the encoding and decoding algorithms of the tANS and \cite{YI2024} for an information-theoretic analysis of the tANS.

\subsection{Asymmetric Encoding-Decoding Scheme (AEDS)}
We now consider a generalization of the tANS. We first generalize the set of internal states to an arbitrary finite set $\mathcal{X}=\{\alpha_1, \alpha_2, \cdots, \alpha_N\}$ with $N=|\mathcal{X}|$.
Furthermore, we generalize the encoding mapping $(s_t, x_t) \mapsto (\beta_t, x_{t-1})$ 
and the decoding mapping $ (\beta_t, x_{t-1}) \mapsto (s_t, x_t)$  to the following general functions.

\vspace{0.2cm}
\noindent {\bf Encoding functions}\footnote{To reduce confusion, different symbols $\hat{x}$ and $x$ are used to represent a state in encoding and decoding, respectively. The symbols `$-$' 
and `$+$' in $F$ functions stand for backward and forward state transitions, respectively.}\\
 \qquad For each $\hat{x}\in\mathcal{X}$, 
\begin{align}
E_{\hat{x}}:\; &\mathcal{S}\to \mathcal{B} \hspace{0.7cm}\text{(Encoding of $s$)},\label{eqE-1}\\
F_{\hat{x}}^{-}:\; &\mathcal{S}\to \mathcal{X}\hspace{0.7cm}\text{(Backward state transition)},\label{eqE-2}
\end{align}

\noindent {\bf Decoding functions}\footnotemark[2]\\
\qquad For each $x\in\mathcal{X}$, 
\begin{align}
D_x:\; &\mathcal{B}^{(D)}_x \to \mathcal{S} \hspace{0.5cm}\text{(Decoding of $s$)},\label{eqD-1}\\
F_x^{+}:\; &\mathcal{B}^{(D)}_x \to \mathcal{X}\hspace{0.5cm}\text{(Forward state transition)},\label{eqD-2}
\end{align}
where $\mathcal{B}=\{0,1\}^*$ includes a null sequence $\lambda$ with length 0,
and a set of codewords $\mathcal{B}^{(D)}_x \subset\mathcal{B}$ is 
defined by $\mathcal{B}^{(D)}_x=\{E_{\hat{x}}(s) \;|\; \hat{x}\in\mathcal{X}, s\in\mathcal{S}$ $\text{ such that } x=F_{\hat{x}}^{-}(s)\}$.
We impose the prefix-free condition on $\mathcal{B}^{(D)}_x$ for each $x\in\mathcal{X}$
to ensure instantaneous and unique decoding.
Then, a set of functions\footnote{When one of
$\{E_{\hat{x}}, F_{\hat{x}}^{-} |\, \hat{x}\in\mathcal{X}\}$ and $\{D_{x}, F_{x}^{+}|\, x\in\mathcal{X}\}$ is given, the other is determined uniquely.}
 $(\{E_{\hat{x}}, F_{\hat{x}}^{-} |\, \hat{x}\in\mathcal{X}\},
\{D_{x}, F_{x}^{+}|\, x\in\mathcal{X}\})$
is called an {\em AEDS (Asymmetric Encoding-Decoding Scheme)}.

\begin{figure}[t]
 \begin{center}
   \includegraphics[width=5cm]{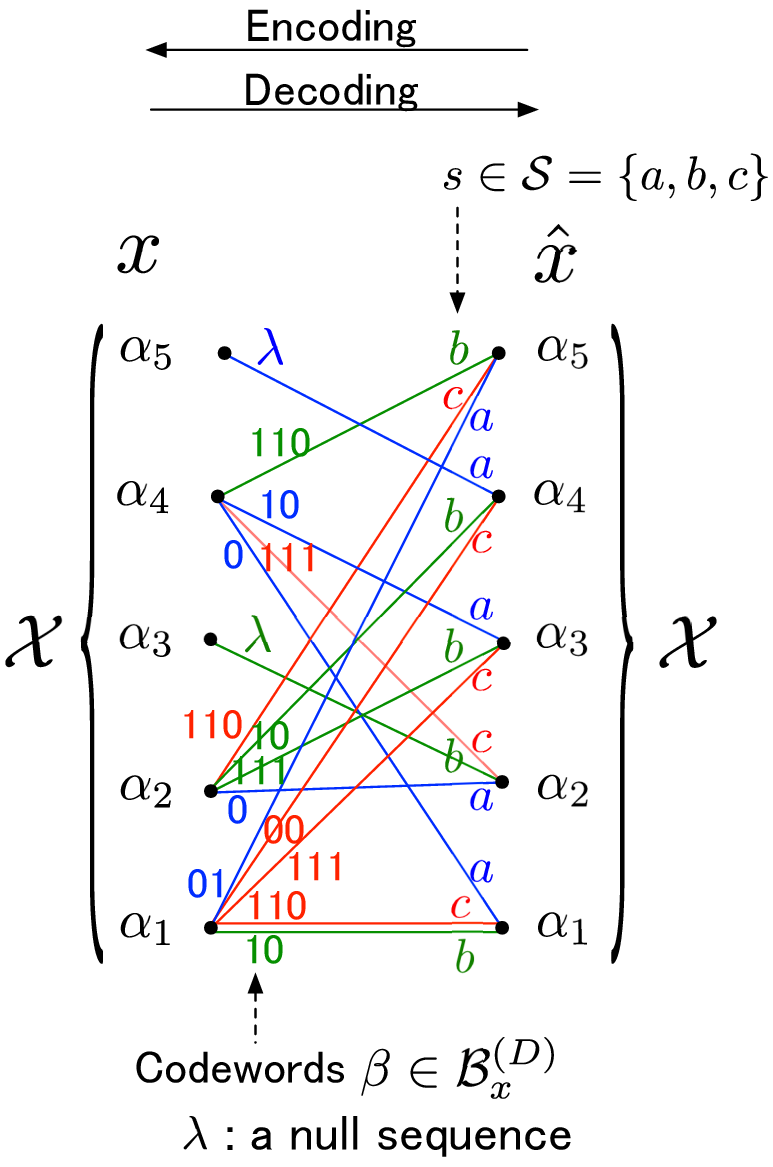}
    \caption{An example of an AEDS for $\mathcal{S}=\{a, b, c\}$, $\mathcal{X}=\{\alpha_1, \alpha_2, \cdots, \alpha_5\}$, and $N=5$.}
      \label{fig1}
 \end{center}
 \vspace*{-0.5cm}
\end{figure}

The encoding and decoding functions $(\{E_{\hat{x}}, F_{\hat{x}}^{-} |\, \hat{x}\in\mathcal{X}\},
\{D_{x}, F_{x}^{+}|\, x\in\mathcal{X}\})$
can be expressed  in an easy-to-understand manner 
using a state transition diagram as shown in Fig.~\ref{fig1}, where $\mathcal{S}=\{a,b,c\}$, $\mathcal{X}=\{\alpha_1, \alpha_2,\cdots, \alpha_5\}$, and $N=5$.
In the figure, each node on the right (resp.~left) side represents a state $\hat{x}\in\mathcal{X}$ in encoding
 (resp.~$x\in\mathcal{X}$ in decoding).
Each line between $x$ and $\hat{x}$ 
represents a state transition, which occurs from the right to the left (resp.~the left to the right)
in encoding (resp.~decoding). Each line has a source symbol $s\in\mathcal{S}$ 
and its codeword  $\beta\in\mathcal{B}$, which correspond to the state transition.

The encoding and decoding procedures of the AEDS can be represented as follows.

\vspace{0.2cm}
\noindent {\bf Encoding procedure of AEDS}
\begin{enumerate}
\item[1.] Set a data sequence $\bm{s}=s_1s_2\cdots s_n$, and select $\hat{x}_n\in\mathcal{X}$ arbitrarily.
\item[2.] Repeat \eqref{eq2B-1}--\eqref{eq2B-2-2} in backward order, starting from $t=n$ down to $t=1$.
\begin{align}
\beta_t&= E_{\hat{x}_t}(s_t), \label{eq2B-1}\\
x_{t}&= F_{\hat{x}_t}^{-}(s_t), \label{eq2B-2}\\
\hat{x}_{t-1}&=x_{t}.       \label{eq2B-2-2} 
\end{align}
\item[3.] The codeword sequence of $\bm{s}$ is given by $\hat{x}_0 \beta_1 \beta_2 \cdots \beta_n$.
\end{enumerate}

\vspace{0.3cm}
\noindent {\bf Decoding procedure of AEDS}
\begin{enumerate}
\item[1.] Set a codeword sequence $\hat{x}_0 \beta_1 \beta_2 \cdots \beta_n$.
\item[2.] Repeat \eqref{eq2B-3-0}--\eqref{eq2B-4} from $t=1$ to $n$ in forward order.
\begin{align}
x_t&= \hat{x}_{t-1}, \label{eq2B-3-0}\\
s_t&= D_{x_{t}}(\beta_t), \label{eq2B-3}\\
\hat{x}_t&= F_{x_{t}}^{+}(\beta_t). \label{eq2B-4}
\end{align}
\item[3.] The decoded data sequence is given by $s_1 s_2 \cdots s_n$.
\end{enumerate}

For instance, for $\bm{s}=s_1s_2s_3s_4=cbba$ and $\hat{x}_4=\alpha_1$ in Fig.~\ref{fig1},
we have that 
$\beta_4=E_{\hat{x}_4}(a)=0$, $\hat{x}_3=x_4=F_{\hat{x}_4}^{-}(a)=\alpha_4$, 
$\beta_3=E_{\hat{x}_3}(b)=10$, $\hat{x}_2=x_3=F_{\hat{x}_3}^{-}(b)=\alpha_2$,
$\beta_2=E_{\hat{x}_2}(b)=\lambda$, $\hat{x}_1=x_2=F_{\hat{x}_2}^{-}(b)=\alpha_3$,
$\beta_1=E_{\hat{x}_1}(c)=111$, $\hat{x}_0=x_1=$ $F_{\hat{x}_1}^{-}(c)=\alpha_1$,
and hence the codeword sequence of $\bm{s}$ is ``$\alpha_1 111 10 0$''.
In this example, we have that $\mathcal{B}^{(D)}_{\alpha_1}=\{00, 01, 10$, $110, 111\}$,
$\mathcal{B}^{(D)}_{\alpha_2}=\{0,10, 110, 111\}$, $\mathcal{B}^{(D)}_{\alpha_3}=\{\lambda\}$,
 $\mathcal{B}^{(D)}_{\alpha_4}=\{0, 10, 110, 111\}$, and $\mathcal{B}^{(D)}_{\alpha_5}=\{\lambda\}$,
 each of which is a prefix-free variable length code, and hence,  $\bm{s}=cbba$ can be easily decoded from 
``$\alpha_1 111 10 0$''.

In the AEDS,  the decoding codeword set $\mathcal{B}^{(D)}_{x}$  satisfies the prefix-free condition for each $x\in\mathcal{X}$.
However, the encoding codeword set $\mathcal{B}^{(E)}_{\hat{x}}=\{E_{\hat{x}}(s)\;|\;s\in\mathcal{S}\}$ does not need to satisfy the prefix-free condition for any  $\hat{x}\in\mathcal{X}$.
Furthermore, the decoding source symbol set $\mathcal{S}^{(D)}_x=\{D_{x}(\beta)\;|\;\beta\in\mathcal{B}^{(D)}_x\}$ does not need to coincide with $\mathcal{S}$ for any  $x\in\mathcal{X}$.
These features, which differ significantly from conventional variable length coding such as Huffman coding,
enable high-performance compression.
On the other hand, we can realize fast encoding and decoding like the Huffman code
since  the AEDS can be implemented without any arithmetic operations 
by using the state transition diagram shown in Fig.~\ref{fig1}.

When we encode $\hat{x}_0$ in a codeword sequence $\hat{x}_0 \beta_1 \beta_2 \cdots \beta_n$ by a fixed length code with $\lceil \ln N \rceil$ bits, 
the code length $L(\bm{s})$ of an i.i.d.~data sequence $\bm{s}=s_1s_2\cdots s_n$ is given by
\begin{align}
L(\bm{s})&=\lceil \lg N\rceil+ \sum_{t=n}^1 l(E_{\hat{x}_t}(s_t)), \label{eql-1}
\end{align}
where $l(\beta_t)$ is the bit length of a codeword $\beta_t=E_{\hat{x}_t}(s_t)$.
When the state transition in $\mathcal{X}$ is Ergodic, i.e.,~irreducible and aperiodic, for
a  source probability distribution $\{p(s)|s\in\mathcal{S}\}$ and $n$ goes to infinity, the average code length $L$ of the AEDS converges to 
\begin{align}
L&=\lim_{n\to\infty}  \frac{1}{n}\sum_{\bm{s}\in{\mathcal{S}^n}}p(\bm{s})L(\bm{s})\nonumber\\
&=\sum_{\hat{x}\in\mathcal{X}}\sum_{s\in\mathcal{S}}  p(s) Q(\hat{x})l(E_{\hat{x}}(s)), \label{eq3-3}
\end{align}
where $Q(\hat{x})$ is the stationary probability of state $\hat{x}\in\mathcal{X}$.
In the following sections, we assume that the state transition is Ergodic, and 
 evaluate the average code length based on \eqref{eq3-3}.

Although the average code length $L$ is represented based on the states $\hat{x}$ of encoding in
\eqref{eq3-3}, $L$ can also be represented based on the states $x$ of decoding as follows.
\begin{align}
L&=\sum_{x\in\mathcal{X}}\sum_{\beta\in\mathcal{B}_x^{(D)}} 
Q(x) \tilde{p}(\beta|x)l(\beta), \label{eq3-3-1}
\end{align}
where $\tilde{p}(\beta|x)$ is given by
\begin{align}
\tilde{p}(\beta|x)=\frac{p(s)Q(\hat{x})}{Q(x)} \text{ for $s=D_x(\beta)$ and $\hat{x}=F^+_x(\beta)$}. \label{eq3-3-2}
\end{align}
We note that $p(s)$ is independent of state $\hat{x}$ in encoding, but $\tilde{p}(\beta|x)$ depends on state $x$ in decoding.
Therefore, we use the simpler representation \eqref{eq3-3} rather than \eqref{eq3-3-1}
in the following sections.

\section{AEDS  based on a code tree}\label{Sec-3}
In this section, we show that an AEDS can be easily constructed from a code tree for a given source.
Let $T$ in Fig.~\ref{figT-1} be a code tree of a prefix-free code for a source distribution $\{p(s)\;|\;
s\in\mathcal{S}\}$. For this code tree, we use the following notation:
\begin{description}
\item{$l_T(s)$:} the codeword length of $s\in\mathcal{S}$ in code tree $T$.
\item{$T_R$:} the right subtree of $T$.
\item{$T_L$:} the left subtree of $T$.
\item{$\mathcal{S}_R$:} the set of source symbols included in $T_R$.
\item{$\mathcal{S}_L$:} the set of source symbols included in $T_L$.
\item{$P_R$:} the probability weight of $T_R$, i.e., $P_R=\sum_{s\in\mathcal{S}_R} p(s)$.
\item{$P_L$:} the probability weight of $T_L$, i.e., $P_L=\sum_{s\in\mathcal{S}_L} p(s)$.
\item{$L_T$:} the average code length of $T$, i.e., $L_T=\sum_{s\in\mathcal{S}}p(s) l_T(s)$.
\item{$L_{T_R}$:} the average  length of $T_R$, i.e., $L_{T_R}=\sum_{s\in\mathcal{S}_R} p(s) (l_T(s)-1)$.
\item{$L_{T_L}$:} the average  length of $T_L$, i.e., $L_{T_L}=\sum_{s\in\mathcal{S}_L} p(s) (l_T(s)-1)$.
\end{description}
Then, we have that $\mathcal{S}_R\cap \mathcal{S}_L=\emptyset$, $\mathcal{S}=\mathcal{S}_R\cup\mathcal{S}_L$, $P_R+P_L=1$, and
\begin{align}
L_T=L_{T_R}+P_R+L_{T_L}+P_L. \label{eqT-1}
\end{align}
We assume $P_R\geq 0.5$ without loss of generality.

\subsection{Type-I AEDS with $N$ states}
Based on a code tree $T$ of Fig.~\ref{figT-1}, we construct an AEDS with $N$ states
as shown in Fig.~\ref{fig-I1}, where $k=\lceil \lg N \rceil$, $\mathcal{B}_{\alpha_1}^{(D)}$ is defined in
Fig.~\ref{fig-I2}, and $\mathcal{B}_{\alpha_j}^{(D)}=T_R$ for $2\leq j \leq N$.
In Fig.~\ref{fig-I2}, $T_{\text{pi}}$ is a code tree of a phased-in code\footnotemark[1]  such that the code tree has $2^k-N$ leaves of $(k-1)$-bit length and $2N-2^k$ leaves of $k$-bit length. 
The codeword composition in encoding is represented on the right side of Fig.~\ref{fig-I1}. 
We call this AEDS {\em Type-I  AEDS}.

\begin{figure}[t]
 \begin{center}
   \includegraphics[width=2.5cm]{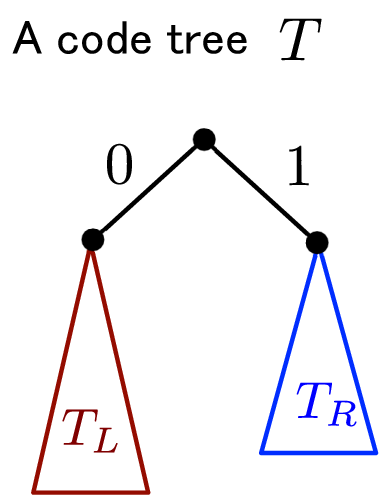}
   \caption{A code tree $T$. Any prefix-free code tree can be used to construct an AEDS. But, a Huffman code tree is used as $T$ in Corollaries 1 and 2.}
      \label{figT-1}
 \end{center}
\end{figure}

\begin{figure}[t]
\centering
\begin{minipage}[b]{0.45\columnwidth}
    \centering
    \includegraphics[width=8cm]{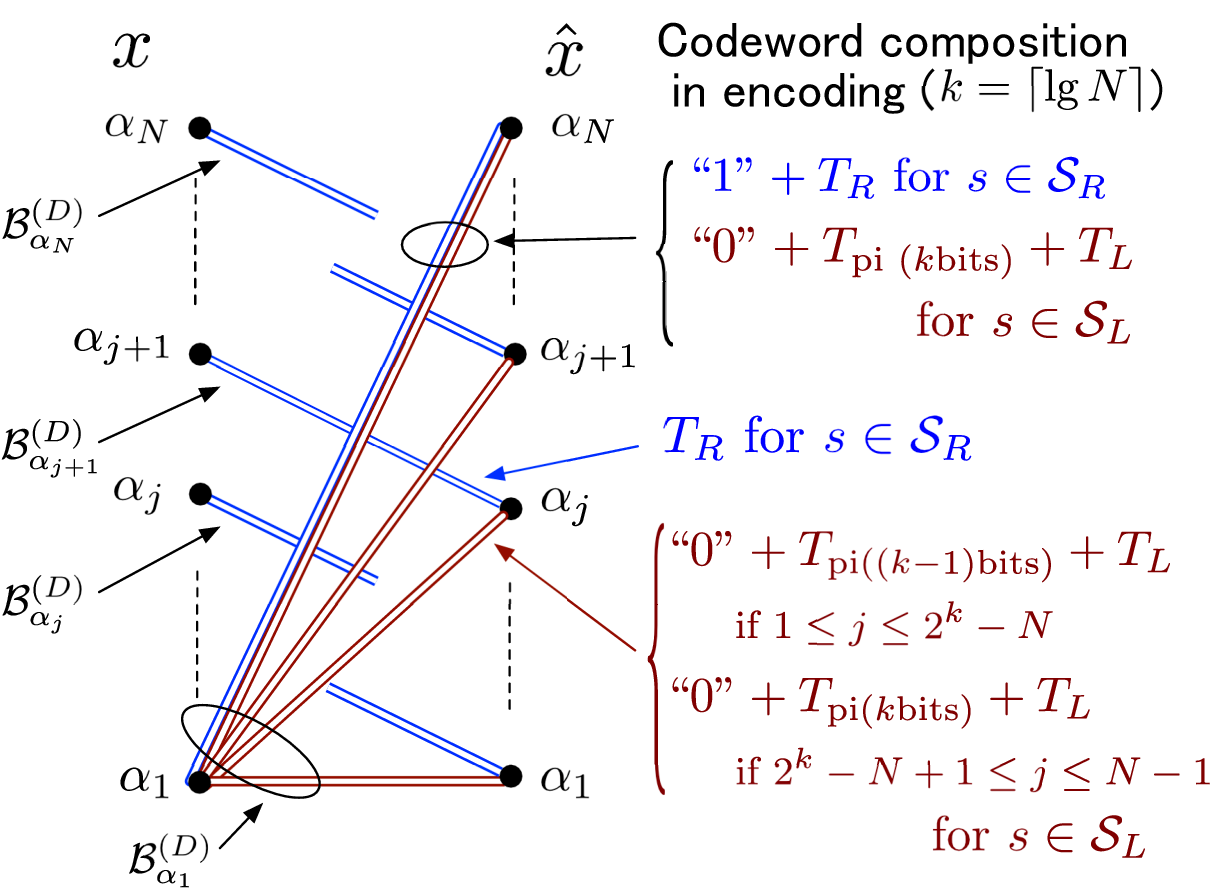}
   \caption{Type-I AEDS with $N$ states based on a code tree $T$,
   where $\mathcal{B}_{\alpha_1}^{(D)}$ is given by Fig.~\ref{fig-I2} and $\mathcal{B}_{\alpha_j}^{(D)}=T_R$ for $2\leq j \leq N$. The right side of the figure indicates the codeword composition in encoding.}
      \label{fig-I1}
\end{minipage}
\begin{minipage}[b]{0.05\columnwidth}
\hspace*{0cm}
\end{minipage}
\begin{minipage}[b]{0.45\columnwidth}
  \centering
    \includegraphics[width=6cm]{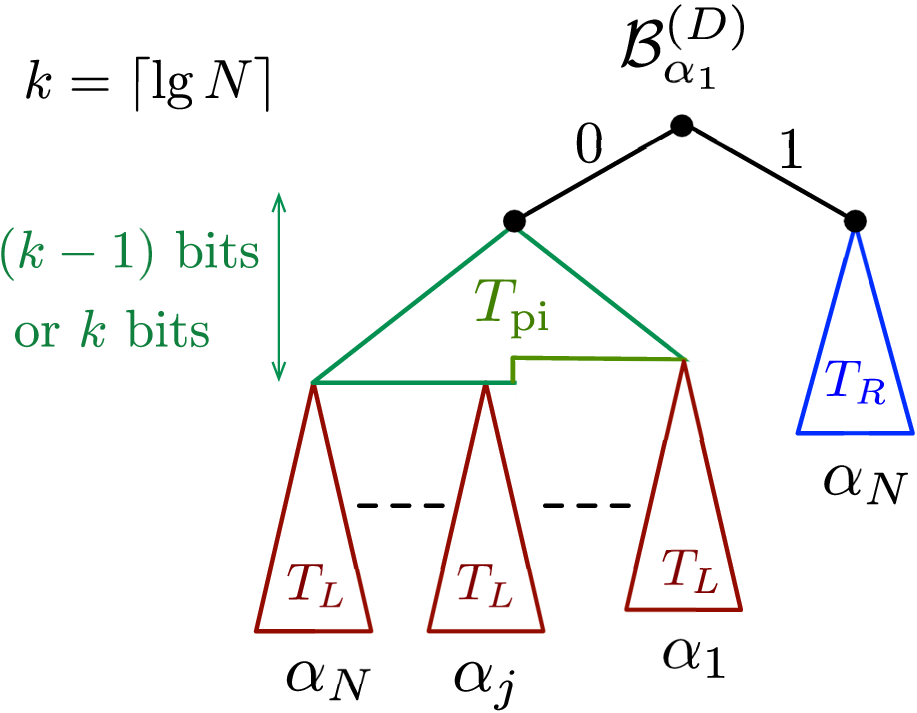}
   \caption{The code tree of $\mathcal{B}_{\alpha_1}^{(D)}$ used in Type-I AEDS.}
   \label{fig-I2}
 \end{minipage}
\end{figure}

Let $L_N^\I$ be the average code length of Type-I AEDS with $N$ states, and we define
the reduction $\delta_N^\I(P_R)$ in average code length attained by Type-I AEDS based on a code tree $T$ with $P_R$ as follows\footnote{We assume that we do not use the AEDS if $L_T-L_N^\I\leq 0$.}.
\begin{align}
\delta_N^\I(P_R) = [L_T-L_N^\I]_0, \label{eq-I1}
\end{align}
where $[u]_0=\max\{u, 0\}$.
Then, the following theorem holds.
\begin{theorem}\label{theorem-I}
For $N\geq 2$ and $k=\lceil \lg N \rceil$, the reduction $\delta_N^\I(P_R)$ of Type-I AEDS is given by
\begin{align}
\delta_N^\I(P_R)=\left[\frac{1-P_R^{N-1}}{1-P_R^N}P_R +\frac{1-P_R^{2^k-N}}{1-P_R^N}(1-P_R)-k(1-P_R)\right]_0, \label{eq-I2}
\end{align}
and $\delta_2^\I(P_R)$ is positive if $P_R>\omega^\I$ where $\omega^\I=(-1+\sqrt{5})/2\approx 0.6180$.
\end{theorem}
The proof of Theorem~\ref{theorem-I} is given in Appendix \ref{App-A}.

\begin{figure}[t]
 \begin{center}
   \includegraphics[width=7cm]{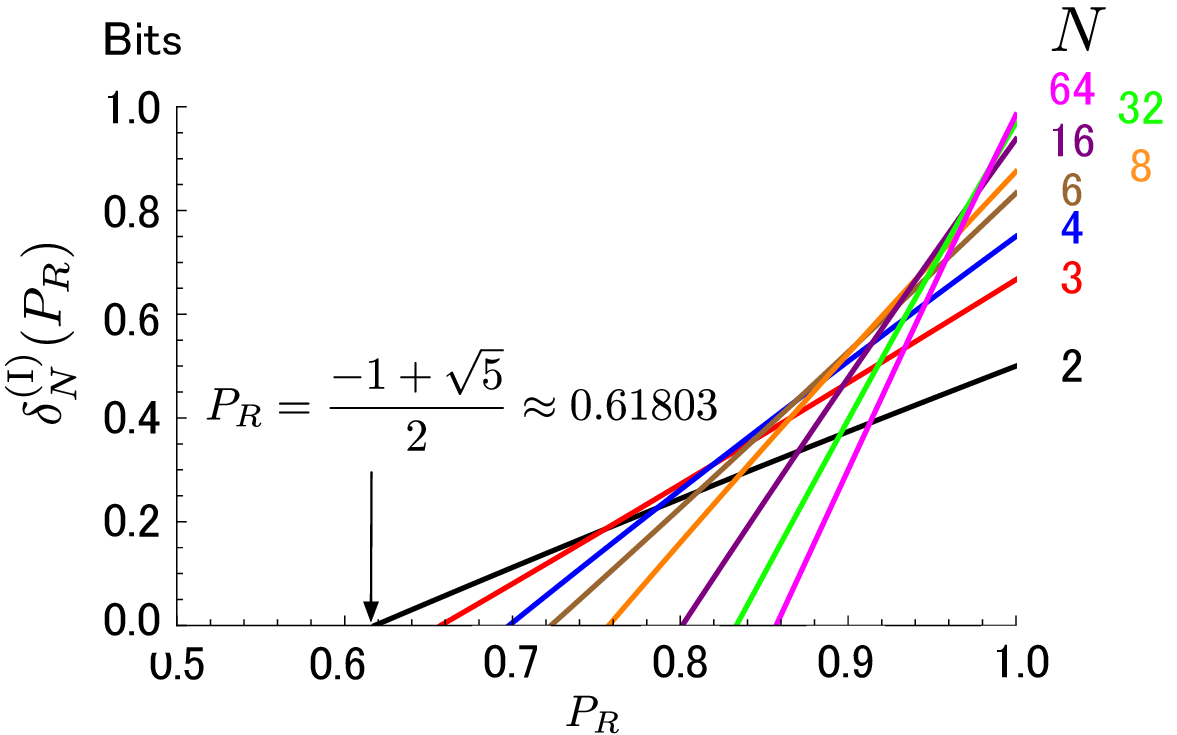}
   \caption{The reduction $\delta_N^\I(P_R)$ attained by Type-I AEDS in the average code length. }
      \label{fig-I3}
 \end{center}
\end{figure}

We note that when $N=2^k$, the phased-in code becomes a fixed $k$-bit-length code
and $\delta_N^\I(P_R)$ is simplified as
\begin{align}
\delta_N^\I(P_R)=&\left[\frac{1-P_R^{N-1}}{1-P_R^N}P_R -k(1-P_R)\right]_0.
\label{eq-I8}
\end{align}
Especially when $N=2$, $\delta_2^\I(P_R)$ is given by
\begin{align}
\delta_2^\I(P_R)&
=\left[\frac{P_R^2+P_R-1}{1+P_R}\right]_0,\label{eq-I7}
\end{align}
which is positive for $P_R>\omega^\I$.
In Fig.~\ref{fig-I3}, $\delta_N^\I(P_R)$ is plotted for several $N$.

\subsection{Type-{\rtwo}~AEDS} 
To reduce the average code length by an AEDS for $P_R\leq \omega^\I\approx0.6180$, we consider another 
type of AEDS with 5 states shown in Fig.~\ref{fig-II1}
based on a code tree $T$ of Fig.~\ref{figT-1}.
In Fig.~\ref{fig-II1}, $\mathcal{B}_{\alpha_1}^{(D)}$ and $\mathcal{B}_{\alpha_3}^{(D)}$ are given in Fig.~\ref{fig-II2}, $\mathcal{B}_{\alpha_2}^{(D)}=T_L$, and $\mathcal{B}_{\alpha_4}^{(D)}=\mathcal{B}_{\alpha_5}^{(D)}=T_R$. The codeword composition in encoding is represented on the right side of Fig.~\ref{fig-II1}. 
We call this AEDS {\em Type-\rtwo~AEDS}.

\begin{figure}[t]
\centering
\begin{minipage}[b]{0.45\columnwidth}
    \centering
   \includegraphics[width=7cm]{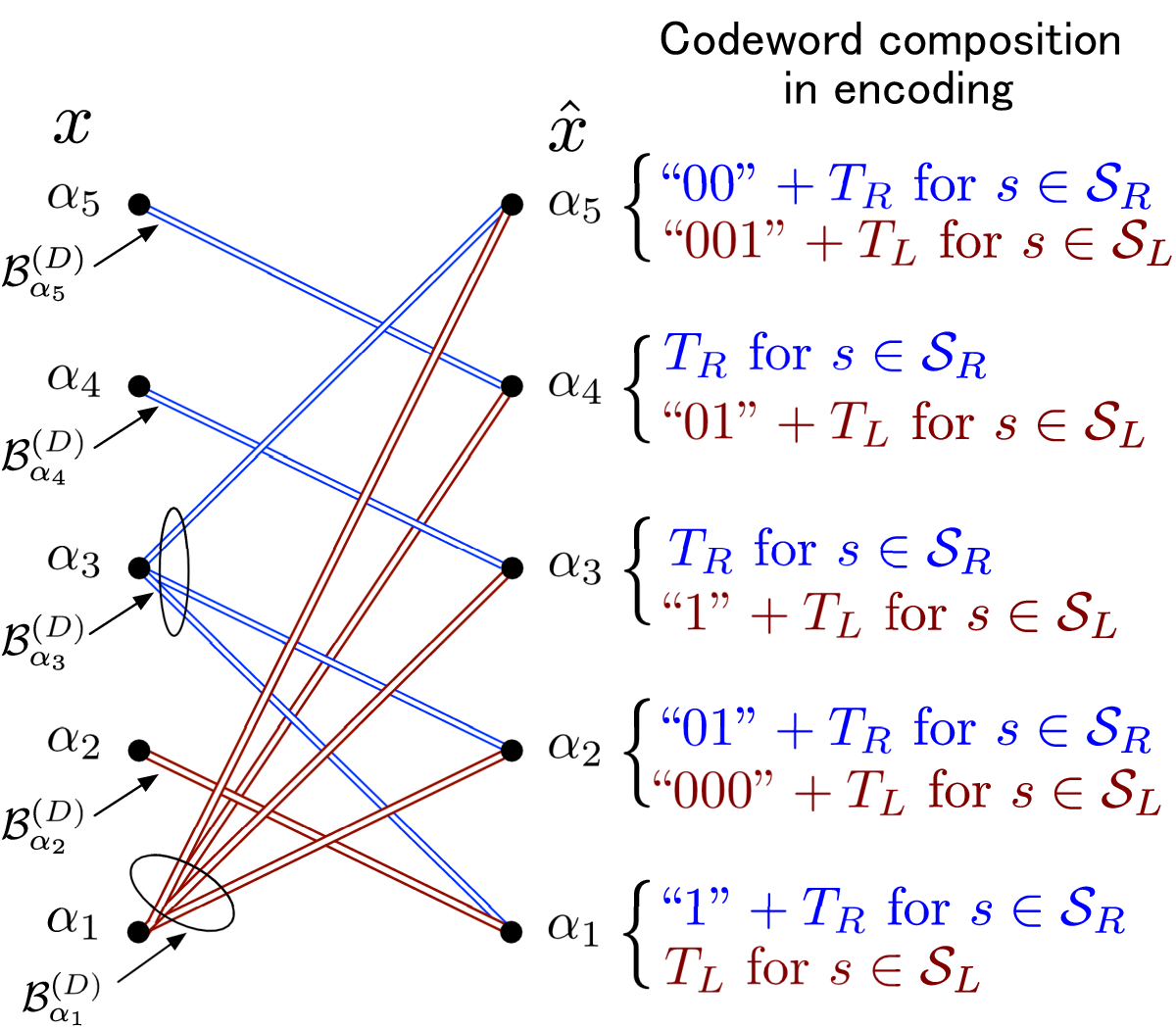}
   \caption{Type-\rtwo~AEDS with 5 states based on a code tree $T$, where $\mathcal{B}_{\alpha_1}^{(D)}$ and $\mathcal{B}_{\alpha_3}^{(D)}$ are given by Fig.~\ref{fig-II2} , $\mathcal{B}_{\alpha_2}^{(D)}=T_L$, and  $\mathcal{B}_{\alpha_4}^{(D)}=\mathcal{B}_{\alpha_5}^{(D)}=T_R$. The right side of the figure indicates the codeword composition in encoding.}
      \label{fig-II1}
\end{minipage}
\begin{minipage}[b]{0.05\columnwidth}
\hspace*{0cm}
\end{minipage}
\begin{minipage}[b]{0.45\columnwidth}
    \centering
  \includegraphics[width=7cm]{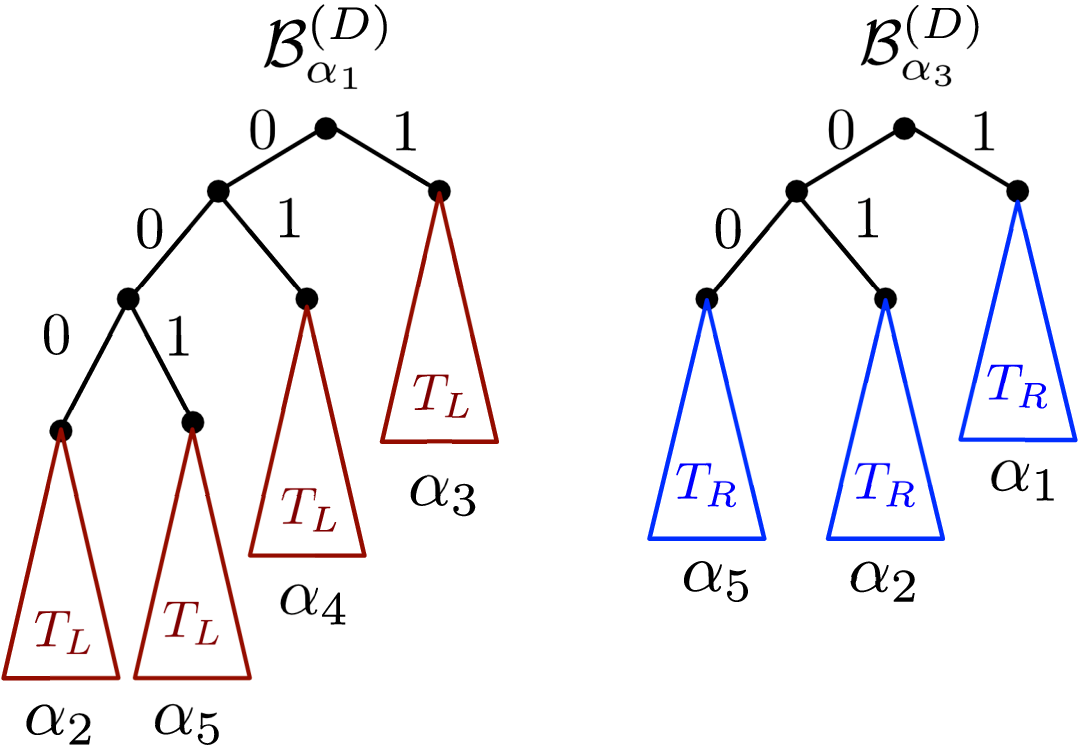}
   \caption{The code trees of $\mathcal{B}_{\alpha_1}^{(D)}$ and $\mathcal{B}_{\alpha_3}^{(D)}$ used in Type-\rtwo~AEDS.}
      \label{fig-II2}
  \end{minipage}
\end{figure}

Let $L^{\II}$ be the average code length of Type-\rtwo~AEDS, 
and we define the reduction $\delta^{\II}(P_R)$ in average code length attained  by Type-{\rtwo}~AEDS based on  a code tree $T$ with $P_R$ as\footnote{We assume that we do not use the AEDS if $L_T-L^{\II}\leq 0$.}\begin{align}
\delta^{\II}(P_R) = [L_T-L^{\II}]_0. \label{eq-II1}
\end{align}
Then, the following theorem holds.
\begin{theorem}\label{theorem-II}
The reduction $\delta^{\II}(P_R)$ of Type-\rtwo~AEDS is given by
\begin{align}
\delta^{\II}(P_R)=\left[\frac{P_R^3-P_R^2+2P_R-1}{(2-P_R)(1+P_R+P_R^2)}\right]_0, \label{eq-II2}
\end{align}
which is positive if $P_R > \omega^{\II}$ where $\omega^{\II}\approx0.56984$ is the real solution of $P_R^3-P_R^2+2P_R-1=0$.
\end{theorem}
Theorem \ref{theorem-II} is proved in Appendix \ref{App-B}.

The reduction $\delta^{\II}(P_R)$ is plotted by a red line in Fig.~\ref{fig-II3}.
We note from the figure that $\delta^{\II}(P_R)$ is larger than $\delta_2^\I(P_R)$
for $0.56984\lesssim P_R \lesssim 0.66536$.

\begin{figure}[t]
 \begin{center}
   \includegraphics[width=7cm]{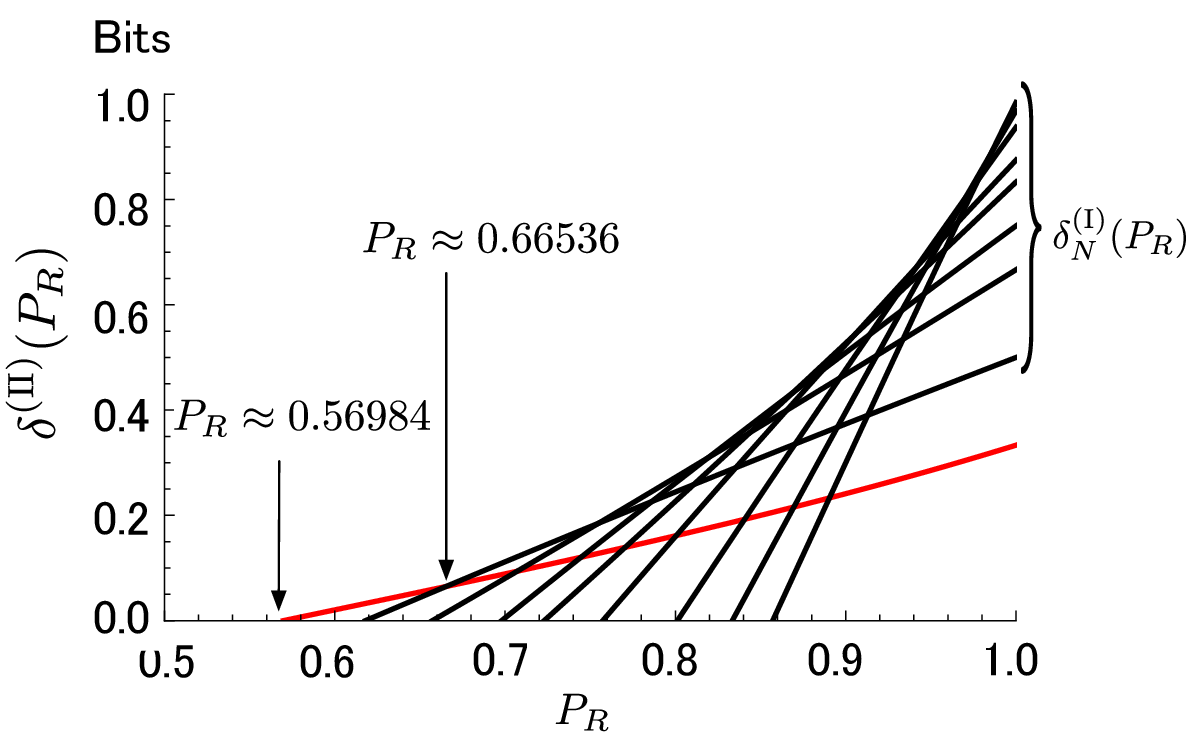}
   \caption{The reduction $\delta^{\II}(P_R)$ attained by Type-\rtwo~AEDS in the average code length.
   Black lines are $\delta_N^\I(P_R)$ shown in Fig.~\ref{fig-I3}. }
      \label{fig-II3}
 \end{center}
\end{figure}

\subsection{Type-I and Type-{\rtwo}~AEDSs based on the Huffman code tree}

In this subsection, we consider the case that Type-I  and Type-\rtwo~AEDSs 
are constructed based on the Huffman code tree for a given source.
Let $L_H$ be the average code length of the Huffman code. 
Then, the next corollary holds from Theorems \ref{theorem-I} and \ref{theorem-II}.
\begin{corollary}\label{Corollary-1}
Assume that the Huffman code tree has $P_R$ for a given source probability distribution.
Then the average code lengths $L_N^\I$ of Type-I AEDS and $L^{\II}$of Type-{\rtwo}~AEDS based on the Huffman code tree are obtained by
\begin{align}
L_N^\I= L_H- \delta_N^\I(P_R), \label{eq-C1}\\
L^{\II}=L_H-\delta^{\II}(P_R), \label{eq-C2}
\end{align}
where $\delta_N^\I(P_R)$ and $\delta^{\II}(P_R)$ are given by \eqref{eq-I2} and \eqref{eq-II2}, respectively.
If the Huffman code tree satisfies $P_R> \omega^\I\approx0.6180$ (resp.~$P_R>\omega^{\II}\approx0.56984$), 
then Type-I AEDS with 2 states (resp.~Type-\rtwo~AEDS with 5 states) attains a shorter average code length than the Huffman code.
\end{corollary}
\begin{proof}
\eqref{eq-C1} and \eqref{eq-C2} are obtained from Theorems \ref{theorem-I} and \ref{theorem-II}
by substituting $L_H$ into $L_T$ in \eqref{eq-I1} and \eqref{eq-II1}.
\end{proof}

Let $p_1, p_2, \cdots, p_{|\mathcal{S}|}$ denote probabilities arranged in descending order.
For example, consider a source with $|\mathcal{S}|=6$ and  $\{p_1=0.35$,
$p_2=p_3=p_4=0.15$, $p_5=p_6=0.1\}$.
In this case,  the Huffman code tree  has $P_R=0.65$ which is larger than $\omega^\I$.
Hence, Type-I AEDS with 2 states (resp.~Type-\rtwo~AEDS)
can achieve an average code length that is $\delta_2^\I(0.65)\approx 0.044$ 
(resp.~$\delta^{\II}(0.65)\approx 0.0544$) shorter than the Huffman code.
If $p_1\geq 0.5$, we have $P_R=p_1$. Hence, for any source satisfying $p_1>\omega^\I\approx 0.6180$ (resp.~$p_1>\omega^{\II}\approx0.56984$), Type-I AEDS with two states (resp.~Type-II AEDS) can attain shorter average code length than the Huffman code.

Next we consider the worst-case redundancy of the AEDS for a given $p_1\geq 0.5$.
It is shown in \cite{Gallager} that the worst-case redundancy $\mu_H(p_1)$ of the Huffman code
is given by
\begin{align}
\mu_H(p_1)&=\sup_{p_2,\cdots,p_{|\mathcal{S}|}} [L_H-H(p)] \label{eq-C3-0}\\
&=2-p_1-h(p_1), \label{eq-C3}
\end{align}
where $H(p)=-\sum_{i=1}^{|\mathcal{S}|} p_i\lg p_i$ is the source entropy and $h(u)=-u\lg u-(1-u)\lg (1-u)$ is the binary entropy function.
On the other hand, the worst-case redundancies of Type-I and Type-{\rtwo} AEDSs, 
which are defined by
\begin{align}
\mu_N^\I(p_1)&=\sup_{p_2,\cdots,p_{|\mathcal{S}|}}  [L_N^\I-H(p)],  \label{eq-C4}\\
\mu^{\II}(p_1)&=\sup_{p_2,\cdots,p_{|\mathcal{S}|}} [L^{\II}-H(p)], \label{eq-C5}
\end{align}
are determined by the following corollary.
\begin{corollary}\label{Corollary-2}
When Type-I and Type-{\rtwo} AEDSs are constructed based on the Huffman code tree for 
a source probability distribution $\{p_1, p_2, \cdots, p_{|\mathcal{S}|}\}$, the worst-case redundancies
for $p_1\geq 0.5$ are given by
\begin{align}
\mu_N^\I(p_1)&= \mu_H(p_1) - \delta_N^\I(p_1), \label{eq-C6}\\
\mu^{\II}(p_1)&= \mu_H(p_1) - \delta^{\II}(p_1). \label{eq-C7}
\end{align}
\end{corollary}
\begin{proof}
\eqref{eq-C6} (resp.~\eqref{eq-C7}) is obtained by combining \eqref{eq-C1}, 
\eqref{eq-C3-0}, and \eqref{eq-C4} 
(resp.~\eqref{eq-C2}, \eqref{eq-C3-0}, and \eqref{eq-C5}) for $P_R=p_1$.
\end{proof}

$\mu_H(p_1)$, $\mu_N^\I(p_1)$, and $\mu^{\II}(p_1)$
are plotted for $0.5\leq p_1<1$ in Fig.~\ref{fig-C1}.
We note that $\lim_{p_1\to 1}\mu_H(p_1)=1$ but $\lim_{p_1\to 1}\mu_N^\I(p_1)=1/N$ and
$\lim_{p_1\to 1}\mu^{\II}(p_1)=2/3$.

\begin{figure}[t]
 \begin{center}
   \includegraphics[width=7cm]{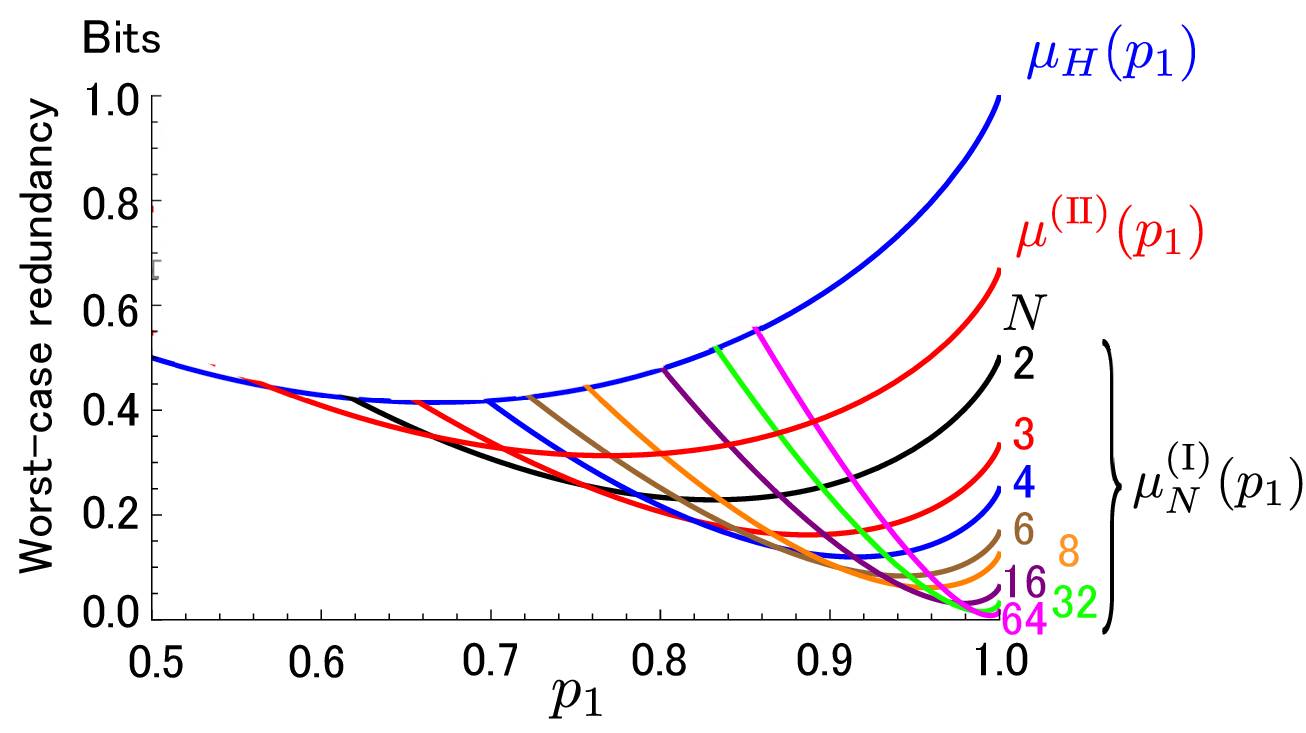}
   \caption{The worst-case redundancies $\mu_H(p_1)$,  $\mu_N^\I(p_1)$, and $\mu_N^{\II}(p_1)$.}
      \label{fig-C1}
 \end{center}
\end{figure}
\begin{remark}\label{remark-1}
When Type-I and Type-{\rtwo}~AEDSs are constructed based on the Huffman code tree, 
they require $N$ states and 5 states, respectively.
However, every $\mathcal{B}_{\alpha_i}^{(D)}$ in the AEDSs consists of $T_R$, $T_L$, or a combination of $T_R$ and $T_L$ of the Huffman code tree as shown in Fig.~\ref{fig-I3} and Fig.~\ref{fig-II1}. 
Hence, with a little ingenuity, the encoding and decoding of these AEDSs can be implemented with nearly the same complexity in space and time as the Huffman coding even if $N$ is large.
\end{remark}

\begin{remark}\label{remark-2}
If a codeword set $\mathcal{B}_x^{(D)}$ in Type I or Type-\rtwo~AEDSs
is not optimal, i.e., not the Huffman code, 
for the probability distribution $\{\tilde{p}(\beta|x) \;|\; \beta\in\mathcal{B}^{(D)}_x\}$ defined by \eqref{eq3-3-2}, 
the average code length of Type-I or Type-\rtwo~AEDSs can be further improved by modifying  
$\mathcal{B}_x^{(D)}$ to the Huffman code of $\{\tilde{p}(\beta|x) \;|\; \beta\in\mathcal{B}^{(D)}_x\}$. But, this improvement increases the space complexity because a different  $\mathcal{B}_x^{(D)}$ must be stored at each state $x\in\mathcal{X}$.
\end{remark}

\subsection{Application of Type-I and Type-{\rtwo}~AEDSs to uniform probability distributions}

The average code length $L_H$ of the Huffman code satisfies $L_H\leq L_T$ for any prefix-free code tree $T$. On the other hand,  $\delta_N^\I(P_R)$ is a monotonically increasing function of $P_R$ for $0.5\leq P_R<1$.
Hence, it may be possible to reduce the average code length $L_N^\I(P_R)=L_T-\delta_N^\I(P_R)$
further by using a code tree $T$ with larger $P_R$ than the Huffman code tree.
In this subsection, we show that such reduction is possible for sources with uniform probability distributions.

Assume that $|\mathcal{S}|=M$ and $p(s)=1/M$ for every $s\in\mathcal{S}$, where $M\geq 2$ is an integer.
Then, for $\kappa=\lceil \lg M\rceil$, the Huffman code\footnote{The Huffman code for a uniform probability distribution is a phased-in code \cite{phasedin}\cite{phasedin-2}.}  has
$(2^{\kappa} -M)$ codewords of $(\kappa-1)$ bits and $(2M-2^{\kappa})$ codewords of $\kappa$ bits. 
Hence,  the average code length $L_{H}(M)$ and the redundancy $\mu_{H}(M)$ of the Huffman code 
for a given $M$ are obtained by
\begin{align}
L_{H}(M)&=\kappa-\frac{2^{\kappa}-M}{M} \nonumber\\
&=\kappa+1 -\frac{2^{\kappa}}{M}, \label{eq-D1}\\
\mu_{H}(M)&=L_{H}(M)-\lg M\nonumber\\
&=\kappa+1 -\frac{2^{\kappa}}{M}-\lg M. \label{eq-D2}
\end{align}
Let $M_R$ (resp.~$M_L$) represent the number of source symbols included in $T_R$ (resp.~$T_L$)
of a Huffman code tree. 
To maximize $P_R$ of the Huffman code tree, we make $M_R$ as large as possible. Then, $P_R$ of the
Huffman code tree, say $P_{R,H}(M)$, is given by
\begin{align}
P_{R,H}(M)=\left\{\begin{array}{ll}
\frac{2^{\kappa-1}}{M}& \text{if $M\geq 3\times 2^{\kappa-2}$,}\\
\frac{M-2^{\kappa-2}}{M}& \text{if $M\leq 3\times 2^{\kappa-2}$.}
\end{array}\right. \label{eq-D3}
\end{align}
Hence,  Type-I AEDS based on the Huffman code tree has reduction $\delta_N^\I(P_{R,H}(M))$,
and the average code length $L_{N,H}^\I(M)$ and the redundancy $\mu_{N,H}^\I(M)$ are
obtained by
\begin{align}
L_{N,H}^\I(M)&=L_H(M)-\delta_N^\I(P_{R,H}(M)), \label{eq-D4}\\
\mu_{N,H}^\I(M)&=L_{N,H}^\I(M)-\lg M\nonumber\\
&=\mu_H(M)-\delta_N^\I(P_{R,H}(M)).  \label{eq-D5}
\end{align}

For a given $M$, we now consider a code tree $T$ such that $T_R$ and $T_L$ are 
the code trees of phased-in codes, i.e.,~Huffman code trees for uniform distributions with cardinality $M_R$ and $M_L$. Then, the average code length 
$L_T(M,M_R)$ and the right probability weight $P_R$ of the code tree $T$ are given by
\begin{align}
L_T(M,M_R)&=\frac{M_R}{M}(1+L_H(M_R))+\frac{M_L}{M}(1+L_H(M_L))\nonumber\\
 &=1+\frac{M_R}{M}L_H(M_R)+\left(1-\frac{M_R}{M}\right)L_H(M-M_R),\label{eq-D6}\\
 P_R&=\frac{M_R}{M}. \label{eq-D7}
\end{align}
Hence, from \eqref{eq-I1} and \eqref{eq-I2}, the average code length $L^\I_N(M,M_R)$ of
Type-I AEDS based on this code tree $T$ becomes
\begin{align}
L_{N}^\I(M,M_R)=L_T(M,M_R)-\delta_N^\I\left(\frac{M_R}{M}\right). \label{eq-D8}
\end{align}
Since $P_R\geq 0.5$, the minimum average code length $L_{N}^\I(M)$ of Type-I AEDS attained by adjusting $M_R$ is represented by
\begin{align}
L_{N}^\I(M)=\min_{M_R: \lceil M/2\rceil \leq M_R\leq M-1} L_{N}^\I(M,M_R). \label{eq-D9}
\end{align}
Although it is difficult to derive the minimum of \eqref{eq-D9} analytically, it is easy to 
obtain the minimum numerically.
Let $\hat{\delta}_N^\I(M)$ be the reduction of $L_{N}^\I(M)$ compared with $L_H(M)$
and $\hat{\mu}_{N}^\I(M)$ be the redundancy of $L_{N}^\I(M)$, which are defined by
\begin{align}
\hat{\delta}_{N}^\I(M)=&L_H(M)-L_{N}^\I(M), \label{eq-D10}\\
\hat{\mu}_{N}^\I(M)=&L_{N}^\I(M)-\lg M \nonumber\\
=& \mu_H(M)-\hat{\delta}_{N}^\I(M).\label{eq-D11}
\end{align}

First we consider the case of $N=2$. 
In Fig.~\ref{fig-D1}, $\mu_H(M)$, $\delta_2^\I(P_{R,H}(M))$, and $\hat{\delta}_2^\I(M)$ are plotted for
$16\leq M\leq 128$.
As shown in Fig.~\ref{fig-D1}, the characteristics for $2^{\ell-1} \leq M \leq 2^{\ell}$ are nearly identical 
for any $\ell$. Therefore, the following describes the characteristics of $\delta_2^\I(P_{R,H}(M))$ and $\hat{\delta}_2^\I(M)$ specially for $\ell=7$, which corresponds to the range of $64 \leq M \leq 128$.

\begin{figure}[t]
 \begin{center}
   \includegraphics[width=9cm]{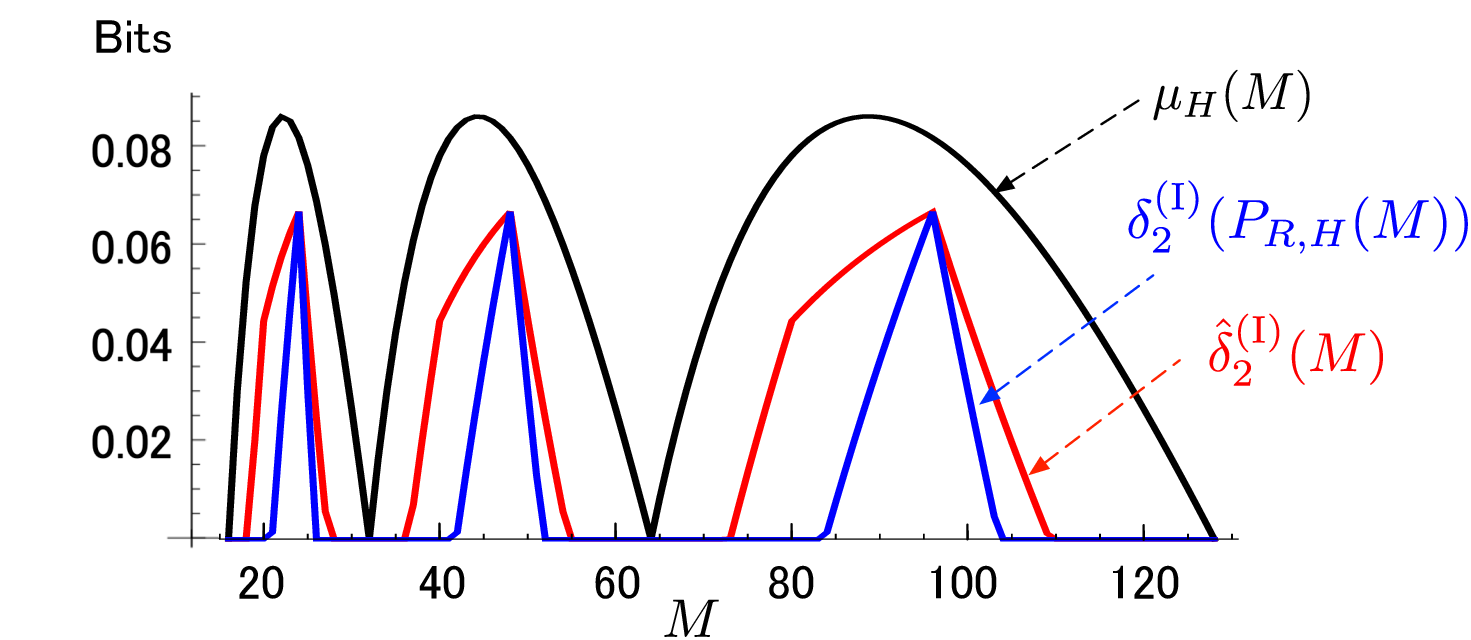}
   \caption{Redundancy $\mu_H(M)$ of Huffman code, reduction $\delta_2^\I(P_{R,H}(M))$ of Type-I AEDS  based on the Huffman code tree, and reduction $\hat{\delta}_2^\I(M)$ of Type-I AEDS based on the optimal code tree. }
      \label{fig-D1}
 \end{center}
\end{figure}

\subsubsection{Characteristics of $\delta_2^\I(P_{R,H}(M))$}
From \eqref{eq-D3} and $3\times 2^5=96$, $P_{R,H}(M)$ is monotonically increasing for $64 \leq M \leq 96$ and monotonically decreasing for $96 \leq M \leq 128$, and it takes the maximum value $P_{R,H}(96)=2/3$ at $M=96$.
Furthermore, we can easily check that $\delta_2^\I(P_{R,H}(M))$ is positive for $84\leq M \leq 103$ since $P_{R,H}(M)>\omega^\I \approx 0.61803$ in this range.   
We note that when $M=96$, the Huffman code tree has $M_R=64$ and $M_L=32$, and the subtrees $T_R$ 
and $T_L$ have zero redundancy, i.e.,~$\mu_H(64)=\mu_H(32)=0$. Therefore, the redundancy $\mu_H(96)$ is caused solely by the bias $P_R$ of the first bit of the Huffman code, and the redundancy can be reduced by Type-I AEDS effectively.

\subsubsection{Characteristics of $\hat{\delta}_2^\I(M)$}
In Table \ref{Table-M_R}, the optimal $M_R$ and $M_L$ to achieve the minimum in \eqref{eq-D9}
 are shown for each $M$ in $73\leq M \leq 109$.
Except for $M=96$, the code tree with the optimal $M_R$ and $M_L$ differs from the Huffman code tree, and $P_R=M_R/M$ of the optimal code tree is larger than $P_{R,H}(M)$.
As a result, we have in $73\leq M \leq 109$ that $\hat{\delta}_2^\I(M)>0$, $\hat{\delta}_2^\I(96)=\delta_2^\I(P_{R,H}(96))$, and $\hat{\delta}_2^\I(M)>\delta_2^\I(P_{R,H}(M))$ for $M\neq 96$ as shown in Fig.~\ref{fig-D1}.

For instance, in the case of $M=80$, the Huffman code tree has $L_H(80)=6.4$,
and Type-I AEDS based on the Huffman code tree cannot reduce the average code length
because it has $P_{R,H}(80)=0.6<\omega^\I\approx 0.61803$ and $\delta_{2}^\I(P_{R,H}(80))=0$. 
On the other hand, the optimal code tree with $M_R=64$ and $M_L=16$ has 
$L_T(80,64)=7\times (64/80)+5\times (16/80)=6.6$, $P_R=M_R/M=0.8$, and  $\delta_{2}^\I(0.8)\approx 0.2444$. Hence,  the average code length of Type-I AEDS based on the optimal code tree
is given by $L_2^\I(80)= L_T(80,64)-\delta_{2}^\I(0.8)\approx6.6-0.2444=6.3556$.
Finally we obtain the reduction $\hat{\delta}_2^\I(80)=L_H(80)-L_2^\I(80)\approx 6.4-6.3556=0.0444$
compared with the Huffman code.

\begin{table}[t]
\begin{center}
\caption{Optimal $M_R$ and $M_L$ for each $M$ in $73\leq M \leq 109$} \label{Table-M_R}
\vspace*{0.1cm}
\begin{tabular}{c||c|c|c|c|c}
$M$& $73, \cdots, 79$ & 80 & $81, \cdots, 95$& 96 & $97, \cdots, 109$ \\ \hline
$M_R$& $M-16$        & 64 &       64           & 64 &    $M-32$ \\ \hline
$M_L$& 16                & 16 &  $M-64$         & 32 &    32 \\
\end{tabular}
\end{center}
\end{table}

It is worth noting that in the case of $M=80$, we have $\mu_H(M_R)=\mu_H(M_L)=0$
for  $M_R=64$ and $M_L=16$ in the same way as the case of $M=96$.
For $73\leq M \leq 109$ other than 80 and 96, the optimal code tree $T$ has either $M_R=64=2^6$,  $M_L=16=2^4$, or $M_L=32=2^5$ as shown in Table \ref{Table-M_R}.

\subsubsection{Characteristics of $\hat{\delta}_N^\I(M)$}
Next, we consider $\hat{\delta}_N^\I(M)$ for $N\geq 2$,
which is plotted for $64\leq M\leq 128$ in Fig.~\ref{fig-D2}.
As shown in Fig.~\ref{fig-I3}, as $N$ increases, the lower bound of $P_R$ satisfying $\delta_N^\I(P_R)>0$
shifts toward 1 and  $\delta_N^\I(P_R)$  becomes larger for $P_R$ close to 1.
On the other hand, Type-I AEDS can attain good performance when $M_R=2^{\ell_R}$ and/or $M_L=2^{\ell_L}$ for some integers $\ell_R$ and $\ell_L$. When $M$ is large in $64\leq M\leq 128$, it is difficult for $P_R=M_R/M$ to satisfy the above two conditions. Hence, $\delta_N^\I(P_R)$ becomes larger for smaller $M$ in $64\leq M\leq 128$ as $N$ increases.

\begin{figure}[t]
 \begin{center}
   \includegraphics[width=6cm]{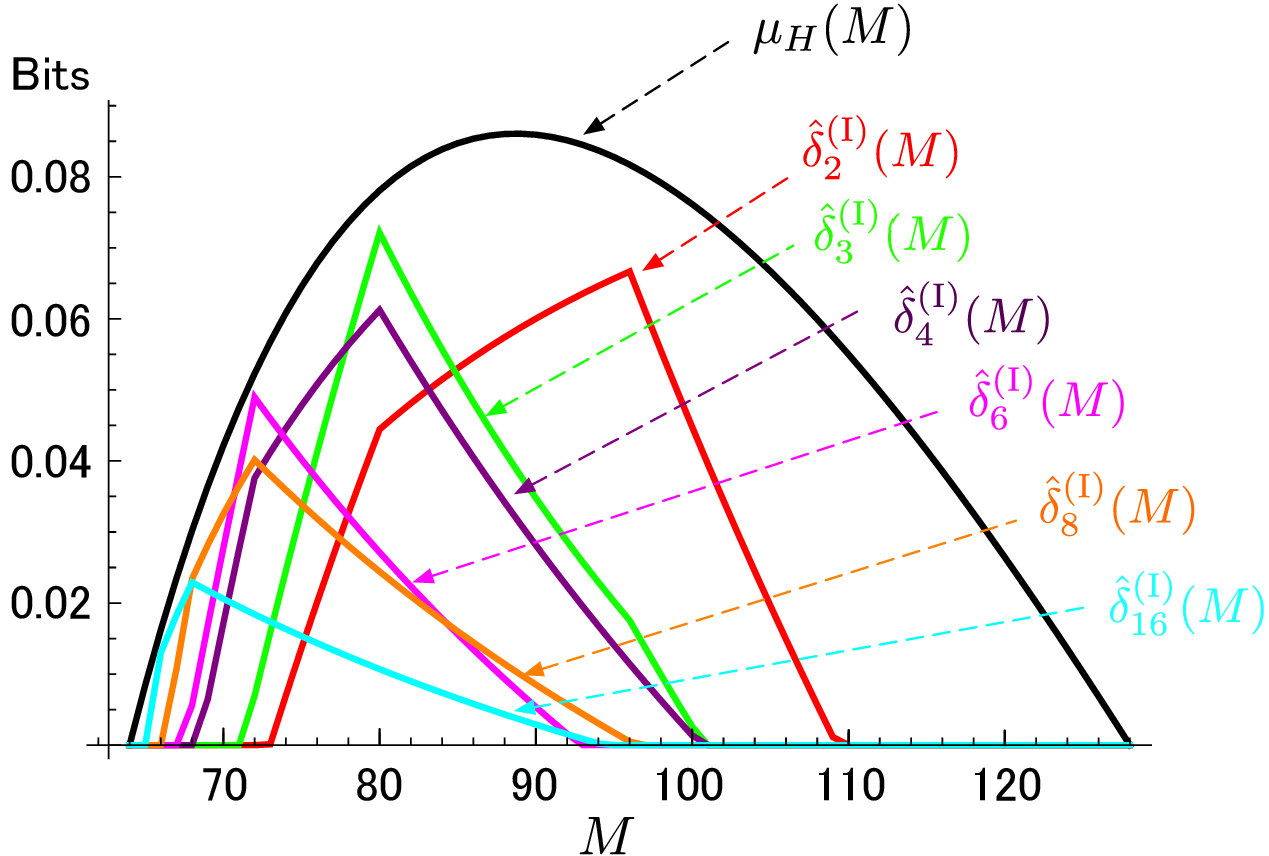}
   \caption{$\hat{\delta}_N^\I(M)$ for $N=2, 3, 4, 6, 8, 16$.}
      \label{fig-D2}
 \end{center}
\end{figure}

For instance, the optimal division of $M=72$ for $N=4,6,8$ is  $(M_R, M_L)=(64, 8)$,
which gives  $P_R=64/72\approx0.889$. 
Similarly the optimal division of $M=68$ for $N=8,16$ is  $(M_R, M_L)=(64,4)$,
which gives $P_R=64/68\approx 0.941$.
In this way, by using an appropriate $N$ in the lower range of $64\leq M\leq 128$,
the reduction $\hat{\delta}_N^\I(M)$ can be made close to the redundancy $\mu_H(M)$ of the Huffman code,
i.e., the redundancy $\hat{\mu}^\I(M)$ of Type-I AEDS can be made small.
From Fig.~\ref{fig-D2}, the optimal Type-I AEDS with an appropriate $N$ can attain
the redundancy $\hat{\mu}^\I(M)$ less than 0.02 (resp.~0.01) for $64\leq M \leq 82$ (resp.~$64\leq M \leq73$).

\subsubsection{Characteristics of $\hat{\delta}^{\II}(M)$}
Next, we consider the case that Type-{\rtwo} AEDS is used instead of Type-I AEDS
for uniform distributions.
Then, in the same way as \eqref{eq-D8}--\eqref{eq-D11} of Type-I AEDS, the reduction $\hat{\delta}^{\II}(M)$ 
and the redundancy $\hat{\mu}^{\II}(M)$ of Type-{\rtwo} AEDS
are defined by using $\delta^{\II}(M_R/M)$ instead of $\delta_N^\I(M_R/M)$ in \eqref{eq-D8}
as follows.
\begin{align}
L^{\II}(M,M_R)&=L_T(M,M_R)-\delta^{\II}\left(\frac{M_R}{M}\right), \label{f}\\
L^{\II}(M)&=\min_{M_R: \lceil M/2\rceil \leq M_R\leq M-1} L^{\II}(M,M_R), \label{eq-D13}\\
\hat{\delta}^{\II}(M)& =L_H(M)-L^{\II}(M), \label{eq-D14}\\
\hat{\mu}^{\II}(M)& =L^{\II}(M)-\lg M \nonumber\\
& = \mu_H(M)-\hat{\delta}^{\II}(M).\label{eq-D15}
\end{align}

\begin{figure}[t]
 \begin{center}
   \includegraphics[width=6cm]{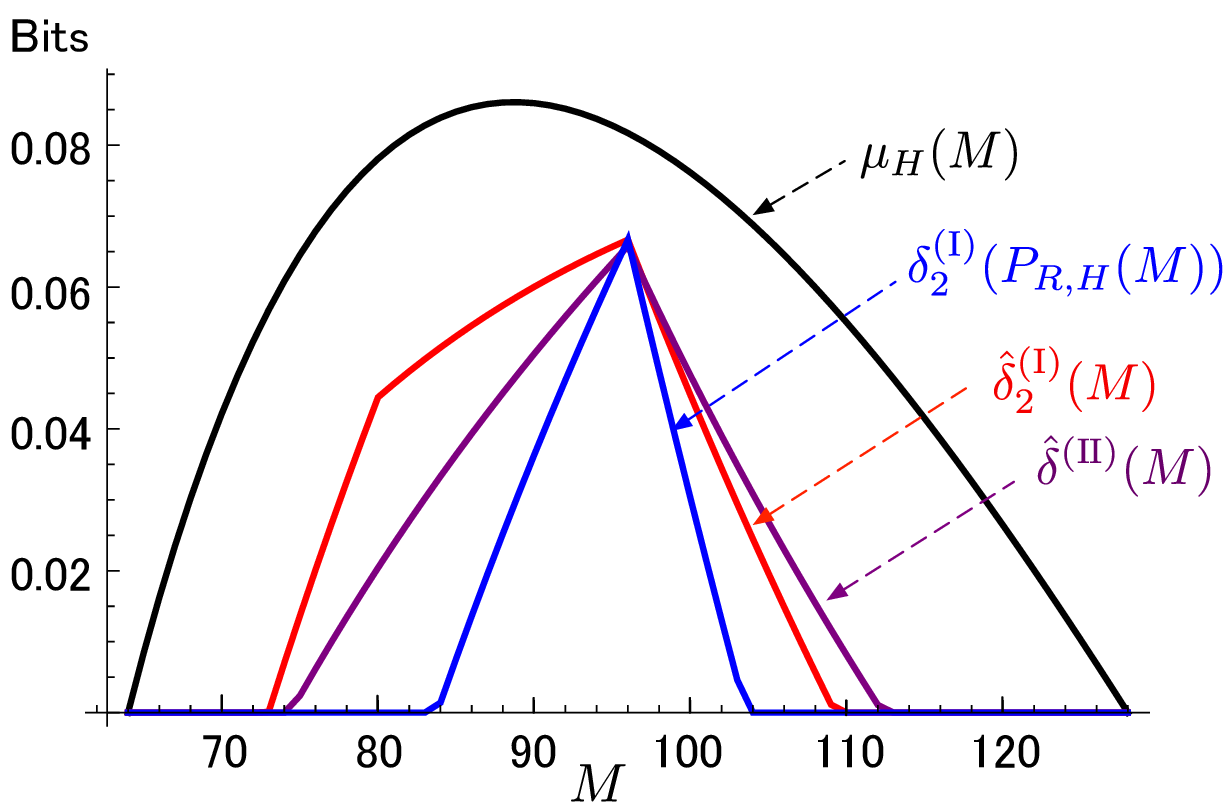}
   \caption{$\hat{\delta}^{\II}(M)$ compared with $\hat{\delta}^{\I}_2(M)$ and $\hat{\delta}^{\I}_2(P_{R,H}(M))$.}
      \label{fig-D3}
 \end{center}
\end{figure}

In Fig.~\ref{fig-D3}, $\hat{\delta}^{\II}(M)$ is plotted with $\mu_H(M)$, $\delta_2^\I(P_{R,H}(M))$,
and $\hat{\delta}_2^\I(M)$ for $64\leq M\leq 128$.
In this range of $M$,  when $M$ increases from 96 to 128, the optimal $M_R$ is fixed at 64 while the optimal $M_L$ increases from 32 to 64, which means that $P_R=M_R/M$ decreases from 2/3 to 1/2.
Since $\delta^{\II}(P_R)>\delta_2^\I(P_R)$ for $0.56984<P_R<0.66536$,
we have $\hat{\delta}^{\II}(M)>\hat{\delta}_2^\I(M)$ for $97\leq M\leq 112$, i.e., $0.6597>P_R=64/M>0.5714$.
On the other hand, for $M\leq96$, i.e, $P_R=64/M\geq 0.6666$, 
we have $\hat{\delta}^{\II}(M)<\hat{\delta}_2^\I(M)$.

\subsection{Type-I and Type-\rtwo~AEDSs for binary sources} \label{subsec-4B}
In this subsection, we treat a binary source with 
$\mathcal{S}=\{a,b\}$, $p(a)=r$, and $p(b)=1-r$.
Without loss of generality, we assume that $0.5\leq r<1$.

For binary sources, the Huffman code tree has $L_H=1$ and $P_R=r$.
Hence, from Corollary \ref{Corollary-1}, 
Type-I and Type-{\rtwo} AEDSs attain the following average code lengths.
\begin{align}
L^\I_N &=1-\delta^\I_N(r), \label{eq-4B-1}\\
L^{\II} &= 1- \delta^{\II}(r). \label{eq-4B-2}
\end{align}
The binary source has redundancy $\mu_S(r)=1-h(r)$, which is also the redundancy of the Huffman code,
where $h(r)$ is the binary entropy function.
On the other hand, from Corollary \ref{Corollary-2}, the redundancies of Type-I and Type-{\rtwo} AEDSs
are given by
\begin{align}
\mu_N^\I(r) &= 1-h(r) - \delta_N^\I(r), \label{eq-4B-3}\\
\mu^{\II}(r) &= 1-h(r) -\delta^{\II}(r), \label{eq-4B-4}
\end{align}
which are plotted in Fig.~\ref{fig3}.
We note from Fig.~\ref{fig3} that Type-I AEDS with an appropriate $N$ or Type-II AEDS can attain redundancy less than 0.0155 for any $r$.

\begin{remark}
In the case of binary sources, we note from Fig.~\ref{fig-I1} that 
the encoding of Type-I AEDS  has $E_{\alpha_i}(a)=\lambda$ for $1\leq i \leq N-1$, where $\lambda$ is a null codeword with length 0. Hence, Type-I AEDS  is closely related to the run length coding
because each codeword $E_{\alpha_i}(b)$ represents the run length $i-1$ of $a$. 
\end{remark}

\begin{figure}[t]
 \begin{center}
   \includegraphics[width=7cm]{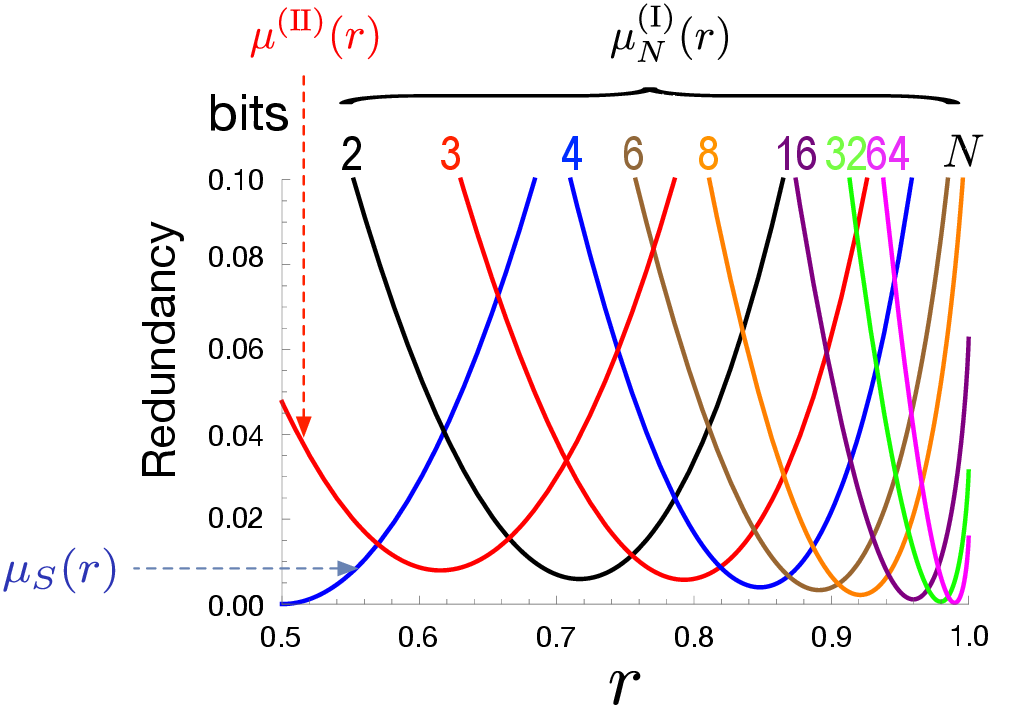}
    \caption{$\mu_S(r)$, $\mu_N^\I(r)$, and $\mu^{\II}(r)$ for binary sources which have $p(a)=r\geq 0.5$ for
    $\mathcal{S}=\{a,b\}$.}
      \label{fig3}
 \end{center}
\vspace*{-0.2cm}
\end{figure}

\section{State-divided AEDS}
As shown in Section \ref{Sec-3}, Type-I and Type-{\rtwo} AEDSs can  easily be constructed and
they can attain a shorter average code length than the Huffman code when 
$P_R>\omega^\I\approx 0.6180$ or $P_R>\omega^{\II}\approx0.56984$ is satisfied in the Huffman code tree.
However, we cannot determine how close the average code length of an AEDS can be made to the source entropy and how an AEDS is closely related to the tANS.
To clarify these problems, 
we consider  a restricted AEDS named a state-divided AEDS (sAEDS) 
 such that $\mathcal{X}$ can be divided based on $s\in\mathcal{S}$ in this section.

For each $s\in\mathcal{S}$, we define $\mathcal{X}_s$ by \eqref{eq2-3}, which
is the set of states reached by  $F_{\hat{x}}^{-}(s)$ from all $\hat{x}\in\mathcal{X}$ in encoding.
\begin{align}
\mathcal{X}_s&=\{F^{-}_{\hat{x}}(s) \;|\; \hat{x}\in\mathcal{X}\}\quad \text{for each $s\in\mathcal{S}$}. \label{eq2-3}
\end{align}
Then, the sAEDS is defined as an AEDS satisfying  \eqref{eq2-4}--\eqref{eq2-6},
where the sets $\mathcal{X}_s$ are mutually disjoint for all $s\in\mathcal{S}$.
\begin{align}
&\mathcal{X}_s\cap \mathcal{X}_{s'}=\emptyset\quad \text{for $s\neq s'$},\label{eq2-4}\\
&|\mathcal{X}_s|\geq 1 \hspace{1.05cm} \text{for each $s\in\mathcal{S}$},\label{eq2-5}\\
&\mathcal{X}=\bigcup_{s\in\mathcal{S}}\mathcal{X}_s. \label{eq2-6}
\end{align}
Since $\mathcal{X}_s$ depends on only $s$, we represent it  as 
$\mathcal{X}_s=\{\s_1, \s_2$, $\cdots, \s_{N_s}\}$ for $N_s=|\mathcal{X}_s|$ using sans-serif font instead of $\alpha_i$.
From \eqref{eq2-4}--\eqref{eq2-6}, we note that $N$ is given by $N=\sum_{s\in\mathcal{S}} N_s$.
An example of an sAEDS is shown in Fig.\!~\ref{fig1-2},
where $\mathcal{S}=\{a, b, c\}$, $N_a=2$, $N_b=2$, $N_c=1$, and $N=5$.

\begin{figure}[t]
 \begin{center}
   \includegraphics[width=5cm]{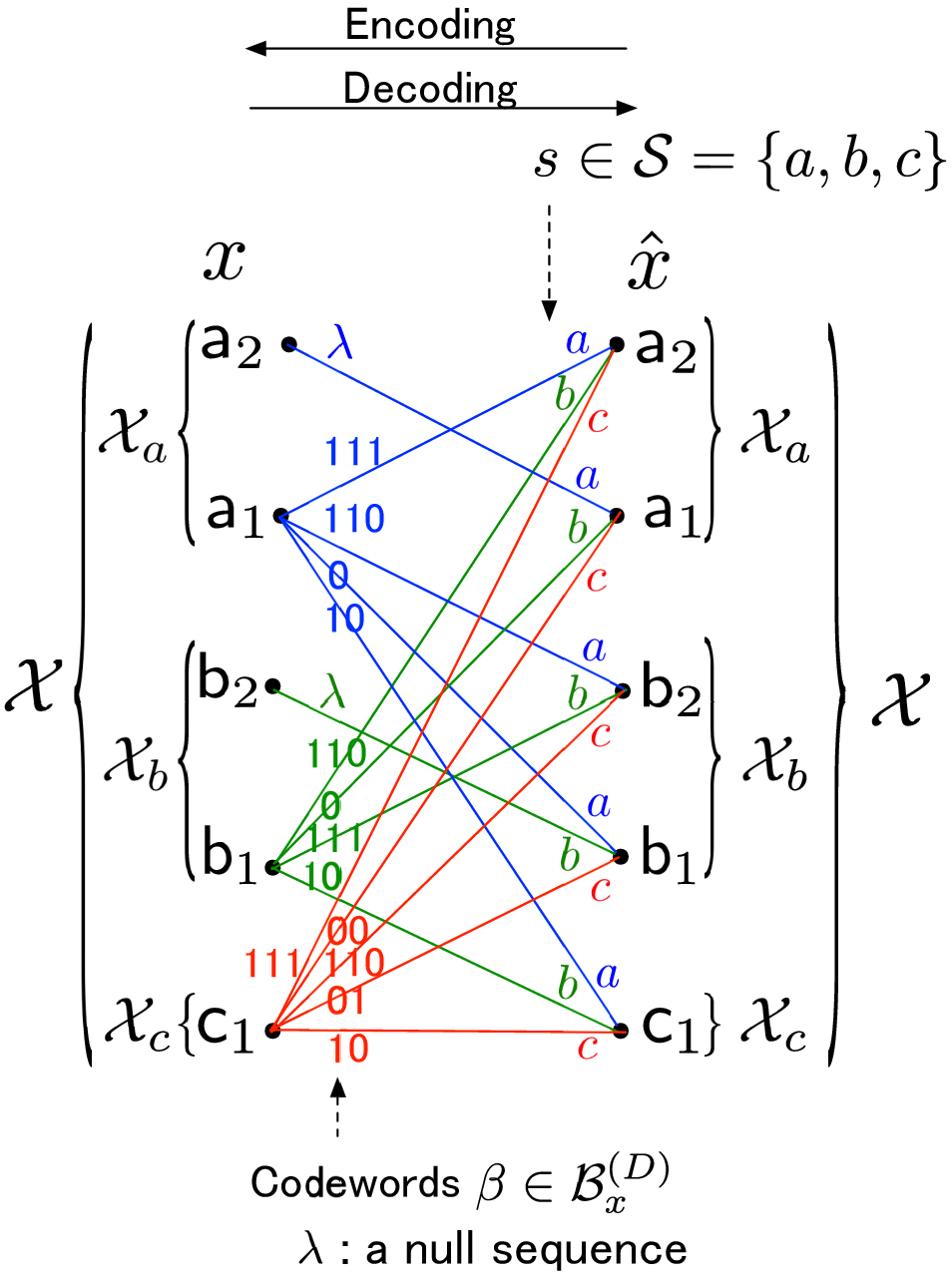}
   \caption{An example of an sAEDS for $\mathcal{S}=\{a, b, c\}$, $N_a=2$, $N_b=2$, and $N_c=1$.}
      \label{fig1-2}
 \end{center}
 \vspace*{-0.5cm}
\end{figure}

In an sAEDS, the set of encoding-decoding functions and state transition functions 
$(\{E_{\hat{x}}, F^-_{\hat{x}}|\hat{x}\in\mathcal{X}\}, \{D_x, F^+_x | x\in\mathcal{X}\})$
is  defined by \eqref{eqE-1}--\eqref{eqD-2} in the same way as a general AEDS. But, 
we note in any sAEDS that if $x\in\mathcal{X}_s$, then $D_{x}(\beta)=s$ for all codewords $\beta\in\mathcal{B}_x^{(D)}$.
Therefore, the decoding function $D_{x}$ can be represented as 
\begin{align}
D_x:\; \emptyset\to \mathcal{S} \hspace{0.5cm}\text{(Decoding of $s$)}\label{eqD-3}
\end{align}
for every $x\in\mathcal{X}_s$ instead of \eqref{eqD-1}.

For each $x\in\mathcal{X}$, we define $\mathcal{F}^{+}_x$ by 
\begin{align}
\mathcal{F}^{+}_x&=\{F^{+}_{x}(\beta) \;|\; \beta\in \mathcal{B}_x^{(D)}\},\label{eq2-7}
\end{align}
which is the set of states reached from $x$ by $F_{x}^{+}(\beta)$ for all $\beta \in\mathcal{B}_x^{(D)}$ in decoding.
For instance, the sAEDS in Fig.~\ref{fig1-2} has the
following $\mathcal{F}_x^{+}$.
\vspace{-0.2cm}
\begin{align*}
\mathcal{F}^{+}_{\asf_1}=\{\asf_2, \bsf_2, \bsf_1, \csf_1\},\quad  \mathcal{F}^{+}_{\asf_2}=\{\asf_1\},\quad 
\mathcal{F}^{+}_{\bsf_1}=\{\asf_2, \asf_1, \bsf_2, \csf_1\},\quad  \mathcal{F}^{+}_{\bsf_2}=\{\bsf_1\},\quad   
\mathcal{F}^{+}_{\csf}=\{\asf_2, \asf_1, \bsf_2, \bsf_1, \csf_1\}.
\end{align*}
Then, since the domain of encoding function $E_{\hat{x}}$ is $\mathcal{S}$ for each $\hat{x}\in\mathcal{X}$, the following relations hold.
\begin{align}
&\mathcal{F}^{+}_x\cap \mathcal{F}^{+}_{x'}=\emptyset
\hspace{0.5cm} \text{if $x\neq x'$, $x\in\mathcal{X}_s, x'\in\mathcal{X}_s$ for each $s\in\mathcal{S}$}, \label{eq2-8}\\
&|\mathcal{F}^{+}_x|\geq 1 \hspace{1.2cm} \text{for each $x\in\mathcal{X}_s$ and each $s\in\mathcal{S}$},\label{eq2-9}\\
&\mathcal{X}=\bigcup_{x\in\mathcal{X}_s}\mathcal{F}^{+}_x\hspace{0.5cm} 
\text{for each $s\in\mathcal{S}$},\label{eq2-10}
\end{align}
which mean that
\begin{align}
N=\sum_{x\in\mathcal{X}_s} |\mathcal{F}^{+}_x| \hspace{0.3cm} 
\text{for each $s\in\mathcal{S}$}.  \label{eq2-10-2}
\end{align}

\begin{remark} \label{rem-2}
Since the tANS satisfies \eqref{eq2-4}--\eqref{eq2-6} \cite{PDPCMM}\cite{PDPCMM2023}, the tANS is included in the class of sAEDS. 
In an sAEDS, any prefix-free code can be used as $\mathcal{B}_x^{(D)}$ for each $x\in\mathcal{X}$.
However, in the case of the tANS,  since the codeword length $k_t$ is given by \eqref{eq2A-1} and \eqref{eq2A-6}, the difference of $k_t$ must be within 1 bit in $\bigcup_{x\in\mathcal{X}_s}\mathcal{B}_x^{(D)}$
for each $s\in\mathcal{S}$.
Hence, the code class of the sAEDS is broader than that of the tANS.
\end{remark}

\subsection{Average code length of sAEDS}
In this subsection, we evaluate the average code length of an sAEDS for i.i.d.~sources 
based on \eqref{eq3-3}.
Let $\{Q(x) \,|\, x\in\mathcal{X}\}$ be the stationary probability distribution on $\mathcal{X}$ 
of an sAEDS for a given source probability distribution $\{p(s) \,|\, s\in\mathcal{S}\}$.
Then, for any $s\in\mathcal{S}$ and any $x\in\mathcal{X}_s$, 
the stationary probability $Q(x)$ satisfies the following relation:
\begin{align}
Q(x)&=p(s)\sum_{\hat{x}\in\mathcal{F}_x^{+}} Q(\hat{x}). \label{eq3-1}
\end{align}
Furthermore, from \eqref{eq2-8}, \eqref{eq2-10}, and \eqref{eq3-1}, we obtain
\begin{align}
\sum_{x\in\mathcal{X}_s} Q(x) &=p(s) \sum_{x\in\mathcal{X}_s} \sum_{\hat{x}\in\mathcal{F}_x^{+}} Q(\hat{x})\nonumber\\
&=p(s) \sum_{\hat{x}\in\mathcal{X}} Q(\hat{x})\nonumber\\
&=p(s). \label{eq3-2}
\end{align}

On the other hand, from \eqref{eq3-3}, \eqref{eq2-8} and \eqref{eq2-10}, 
the average code length $L$ of an sAEDS can be represented as follows.
\begin{align}
L=\sum_{s\in\mathcal{S}}p(s)\sum_{x\in\mathcal{X}_s}\sum_{\hat{x}\in\mathcal{F}_x^{+}}
 l(E_{\hat{x}}(s))Q(\hat{x}). \label{eq3-4}
\end{align}

In a general AEDS, the average code length $L$ is given by \eqref{eq3-3}, and hence
$l(E_{\hat{x}}(s))$ must be designed for the product of $p(s)$ and $Q(\hat{x})$
to realize a small $L$. On the other hand, in the case of sAEDS, we note 
from \eqref{eq3-4} that for each $s\in\mathcal{S}$ and each $x\in\mathcal{X}_s$, $l(E_{\hat{x}}(s))$ 
can be designed based on only $\{Q(\hat{x})|\hat{x}\in\mathcal{F}^+_x\}$. Hence, the theoretical analysis and design of
an sAEDS becomes easier than those of a general AEDS.
It is also worth noting that the stationary probability distribution $\{Q(\hat{x})|\hat{x}\in\mathcal{X}\}$
depends on only $\{\mathcal{F}^+_x |x\in\mathcal{X}_s, s\in\mathcal{S}\}$ and $\{p(s)|s\in\mathcal{S}\}$, 
regardless of what prefix-free codes $\mathcal{B}^{(D)}_x$ 
are used for $\mathcal{F}^+_x, x\in\mathcal{X}$.

\subsection{sAEDS with the same code length as Huffman code}
Assume that for a source probability distribution $\{p(s) \,|\, s\in\mathcal{S}\}$, 
the Huffman code has code length $l_H(s)$ for $s\in\mathcal{S}$, 
the average code length $L_H=\sum_{s\in\mathcal{S}}p(s)l_H(s)$, and the maximum code length
$l_{H}^{\max}=\max_{s\in\mathcal{S}}l_H(s)$.
Now, we consider an sAEDS with $N=2^{l_{H}^{\max}}$ and $N_s=2^{l_{H}^{\max}-l_H(s)}$ for  $s\in\mathcal{S}$.
Then,  for each $s\in\mathcal{S}$, we can divide $\mathcal{X}$ into $\{\mathcal{F}_x^{+}\}$
so that $|\mathcal{F}_x^{+}|=N/N_s=2^{l_H(s)}$ for every $x\in \mathcal{X}_s$.
When we use a fixed length code for $\mathcal{B}_x^{(D)}$ for each $x\in\mathcal{X}_s$, the code length is given by  $l(E_{\hat{x}}(s))=l_H(s)$ for  each $\hat{x}\in\mathcal{X}$. 
Therefore, from  \eqref{eq2-8}, \eqref{eq2-10}, and  \eqref{eq3-4}, the average code length of this sAEDS is given by
\begin{align}
L&=\sum_{s\in\mathcal{S}}p(s)\sum_{x\in\mathcal{X}_s}\sum_{\hat{x}\in\mathcal{F}_x^{+}}
 l_H(s) Q(\hat{x})\nonumber\\
&= \sum_{s\in\mathcal{S}} p(s) l_H(s) \sum_{\hat{x}\in\mathcal{X}}Q(\hat{x})\nonumber\\
&= \sum_{s\in\mathcal{S}} p(s) l_H(s) \nonumber\\
&=L_H. \label{eq4-1}
\end{align}

The above sAEDS is usually not optimal because only fixed length codes are used for all $\mathcal{B}_x^{(D)}$.
Hence the following theorem holds.
\begin{theorem}\label{th-2}
If the Huffman code has the maximum code length ${l_H^{\max}}$ for a given source,
the optimal sAEDS with $2^{l_H^{\max}}$ states can attain 
an average code length shorter than (or at worst equal to) the Huffman code.
\end{theorem}

The above sAEDS is not efficient because it requires $N$ times the space complexity of
the Huffman code. However, Theorem~\ref{th-2} is useful for ensuring that the 
optimal sAEDS is not worse than the Huffman code in the average code length
when we can use $N$ states satisfying $N\geq 2^{l_H^{\max}}$.

\subsection{Upper bounds on the average code length of sAEDS}
Assume that an sAEDS has $\{N_s | s\in\mathcal{S}\}$ and $N=\sum_{s\in\mathcal{S}}N_s$
for a source probability distribution $p=\{p(s)| s\in\mathcal{S}\}$ where $2\leq |\mathcal{S}|<\infty$. 
Then, we define a probability distribution $q=\{q(s) | s\in\mathcal{S}\}$ by $q(s)=N_s/N$, and  represent the source entropy and the relative entropy by $H(p)=-\sum_{s\in\mathcal{S}}p(s)\lg p(s)$ and $D(p\|q)=\sum_{s\in\mathcal{S}}p(s)\lg (p(s)/q(s))$, respectively.

We evaluate the average code length $L$ of an sAEDS 
in several cases such that each $\mathcal{B}^{(D)}_x$ is constructed by a phased-in code.

\vspace{0.1cm}
\subsubsection{Case 1 (All $N/N_s$ are integers)}

 Assume that $M_s=N/N_s$ is an integer for each $s\in\mathcal{S}$.
 Then,  from $|\mathcal{X}_s|=N_s$ and \eqref{eq2-10-2}, 
 we can choose $\mathcal{F}^{+}_x$ so that $|\mathcal{F}^{+}_x|=M_s$ is satisfied for all $x\in\mathcal{X}_s$,
 and hence we can use the phased-in code with $M_s$ codewords for each $\mathcal{F}_x^+$, $x\in\mathcal{X}_s$.  See Fig.~\ref{fig-sAEDS-D1}.  
 Then, the following theorem holds.
 \begin{theorem}\label{th-3}
In the case that $N/N_s$ is an integer for each $s\in \mathcal{S}$, the average code length $L$ of the above sAEDS is bounded by
\begin{align}
L\leq H(p) +D(p\|q) +\sigma, \label{eq5-th1}
\end{align}
where  $\sigma=\lg \lg e + 1-\lg e\approx 0.08607$.
\end{theorem}
Theorem \ref{th-3} is proved in Appendix \ref{App-C}.
 
 \begin{figure}[t]
 \begin{center}
   \includegraphics[width=5cm]{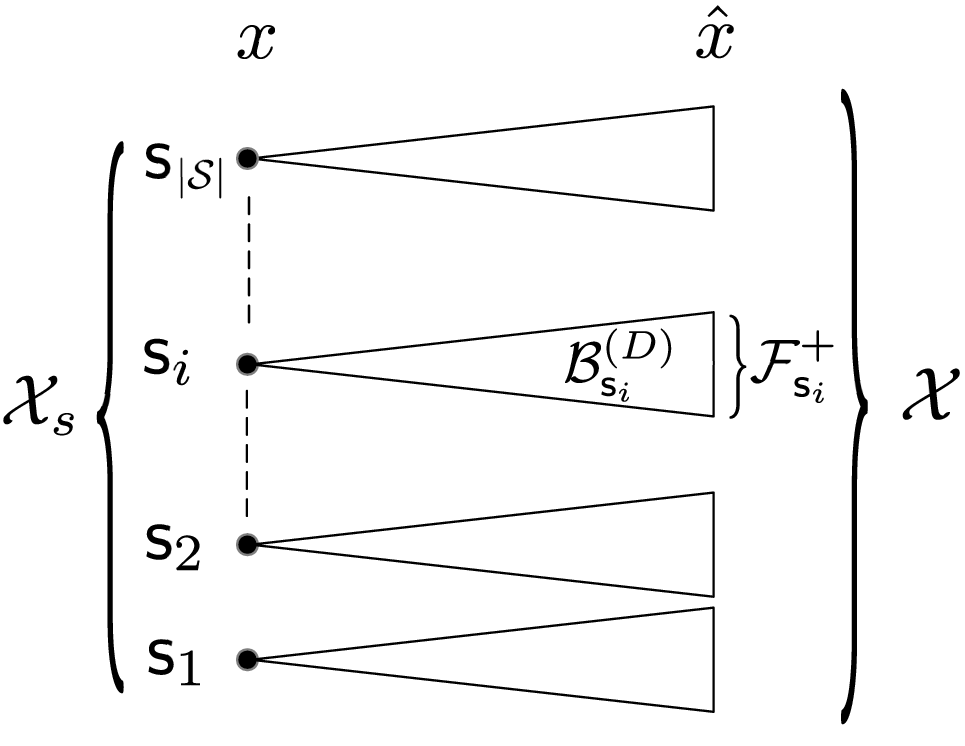}
    \caption{An sAEDS with $N=|\mathcal{X}|$, $N_s=|\mathcal{X}_s|$ and $M_s=N/N_s=|\mathcal{F}_x^+|=|\mathcal{B}^{(D)}_x|$ for $x\in\mathcal{X}_s$, where $\mathcal{B}^{(D)}_x$ is a phased-in code.}
      \label{fig-sAEDS-D1}
 \end{center}
\vspace*{-0.2cm}
\end{figure}

\subsubsection{Case 2: ($N/N_s$ may not be integers)}
Next we consider a general case such that $N/N_s$ may not be an integer.
Letting $M_s=\lfloor N/N_s\rfloor$ and $r_s=N\bmod N_s$ for each $s\in\mathcal{S}$,
we can choose $\{\mathcal{F}_x^{+} \,|\, x\in\mathcal{X}_s\}$ so that it satisfies $|\mathcal{F}_x^{+}|=M_s$ for $N_s-r_s$ states of $x$ and $|\mathcal{F}_x^{+}|=M_s+1$ for $r_s$ states of $x$ in $\mathcal{X}_s$.
See Fig.~\ref{fig-sAEDS-D2}.
Then when we use phased-in codes with $M_s$ and $M_s+1$ codewords for the former
and latter $\mathcal{F}_x^{+}$, respectively, this sAEDS satisfies the following theorem.
\begin{theorem}\label{th-4}
For any positive integers $\{N_s\}$ and $N=\sum_{s\in\mathcal{S}}N_s$,
the average code length $L$ of the above sAEDS is bounded by
\begin{align}
L \leq H(p) + D(p\|q) +\sigma
+ \sum_{s\in\mathcal{S}}p(s)\lg \frac{M_s+\tilde{Q}_{M_s+1}}{N/N_s}, \label{eq5-th2}
\end{align}
where $M_s=\lfloor N/N_s \rfloor$ and $\tilde{Q}_{M_s+1}$ is defined for $\tilde{Q}_s(x)=\sum_{\hat{x}\in\mathcal{F}^+_s} Q(\hat{x})$ by
\begin{align}
\tilde{Q}_{M_s+1} =\sum_{x:\; x\in\mathcal{X}_s,\; |\mathcal{F}_x^{+}|=M_s+1}\tilde{Q}_s(x).\label{eq5-5}
\end{align}
\end{theorem}
Theorem \ref{th-4} is proved in Appendix \ref{App-D}.

 \begin{figure}[t]
 \begin{center}
   \includegraphics[width=5cm]{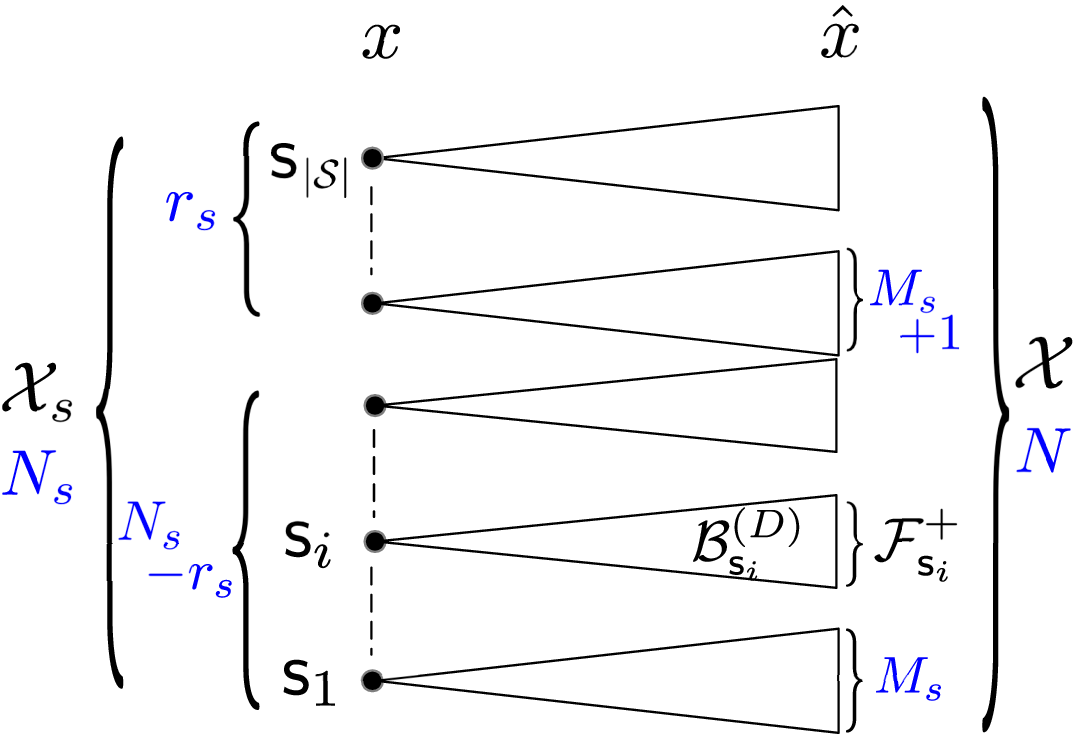}
    \caption{An sAEDS with $N=|\mathcal{X}|$, $N_s=|\mathcal{X}_s|$,  $M_s=\lfloor N/N_s\rfloor$, 
    $r_s= N \bmod N_s$, $|\mathcal{F}_x|=|\mathcal{B}^{(D)}_x|=M_s$ for $N_s-r_s$ states of $x$ and  $|\mathcal{F}_x|=|\mathcal{B}^{(D)}_x|=M_s+1$ for $r_s$ states of $x$ in $\mathcal{X}_s$,
   where $\mathcal{B}^{(D)}_x$ is a phased-in code.}
      \label{fig-sAEDS-D2}
 \end{center}
\vspace*{-0.2cm}
\end{figure}

Compared with \eqref{eq5-th1}, the second term $D(p\|q)$ in \eqref{eq5-th2}  can be made small because $M_s=N/N_s$ does not need to be an integer although \eqref{eq5-th2} has the additional fourth term. 
We note that if $\tilde{Q}_{M_s+1}\leq r_s/N_s$ holds, we have
\begin{align}
\lg \frac{M_s+\tilde{Q}_{M_s+1}}{N/N_s}
\leq \lg\frac{M_s+(r_s/N_s)}{N/N_s}=\lg\frac{M_sN_s+r_s}{N}=\lg\frac{N}{N}= 0. \label{eq5-7}
\end{align}
Hence, if $\tilde{Q}_{M_s+1}\leq r_s/N_s$ holds for every $s\in\mathcal{S}$, 
we have $L\leq H(S) + D(p\|q) +\sigma$. If there is a certain degree of bias in $\{Q(\hat{x}) \,|\, \hat{x}\in\mathcal{X}\}$ and elements $\hat{x}$ with smaller probabilities $Q(\hat{x})$ are included in $\mathcal{F}^+_x$ with $|\mathcal{F}^+_x|=M_s+1$, this condition can be satisfied.

\subsubsection{Case3 ($N$ is a power of 2)}
Lastly, we consider the case that a given $\{N_s \,|\, s\in\mathcal{S}\}$ satisfies $N=\sum_{s\in\mathcal{S}}N_s=2^k$ for a positive integer $k$.
In this case,  $\mathcal{X}$ can be divided into $\{\mathcal{F}^+_x \,|\, x\in\mathcal{X}_s\}$ so that   a fixed length code can be used for each $\mathcal{F}^+_x$ as follows.
For $k_s=\lceil \lg N_s\rceil$, let $N_{s,1}$ (resp.~$N_{s,2}$) be the number of $\mathcal{F}^+_x$ that uses a fixed length code with $k-k_s+1$ bits (resp.~$k-k_s$ bits).
Then, if we set $N_{s,1}=2^{k_s}-N_s$ and $N_{s,2}=2N_s-2^{k_s}$, we can 
satisfy both \eqref{eq5-8} and \eqref{eq5-9}. 
\begin{align}
N_{s,1}+N_{s,2}&=2^{k_s}-N_s+2N_s-2^{k_s}\nonumber\\
&=N_s. \label{eq5-8}\\
N_{s,1}2^{k-k_s+1}+N_{s,2}2^{k-k_s}&=(2^{k_s}-N_s)2^{k-k_s+1}+(2N_s-2^{k_s})2^{k-k_s}\nonumber\\
&= 2^{k+1}-N_s2^{k-k_s+1}+2N_s2^{k-k_s}-2^k\nonumber\\
&=2^k\nonumber\\
&=N. \label{eq5-9}
\end{align}
See Fig.~\ref{fig-sAEDS-D3}. 
Since \eqref{eq5-8} and \eqref{eq5-9} mean that \eqref{eq2-10-2} is satisfied, 
we can always use a fixed length code for each $\mathcal{F}^+_x$, $x\in\mathcal{X}_s$ for any given $N_s$. 
Then, the following theorem holds.
 \begin{figure}[t]
 \begin{center}
   \includegraphics[width=5cm]{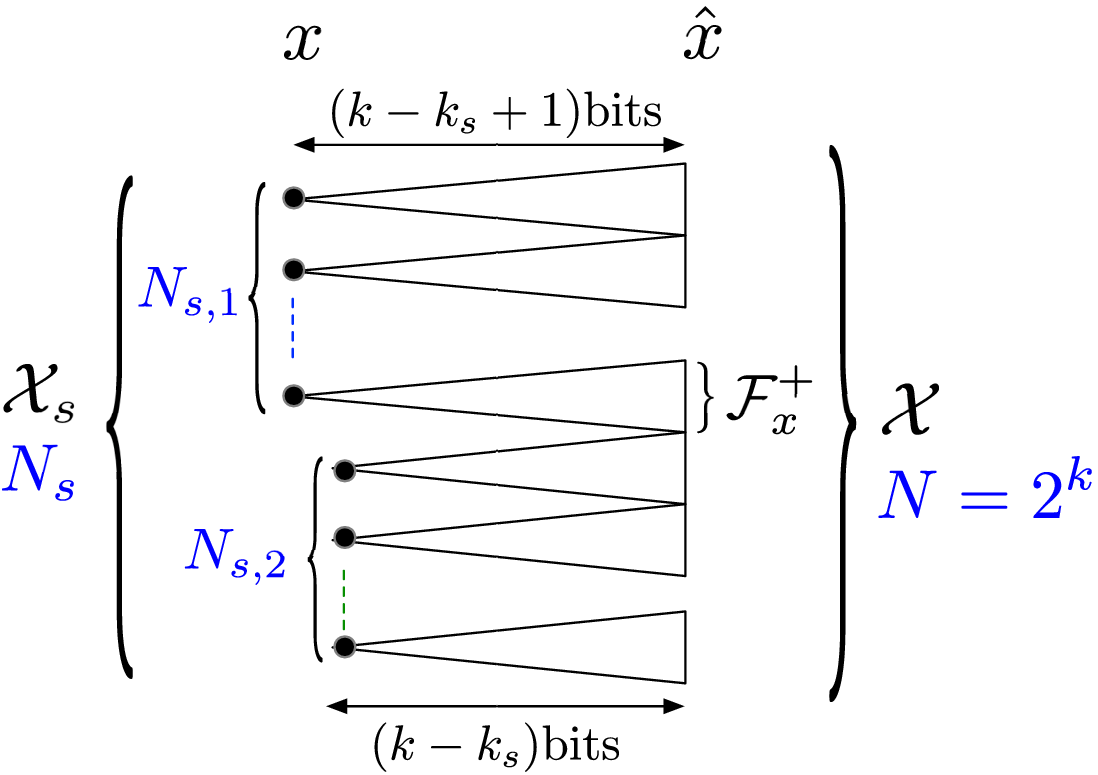}
    \caption{An sAEDS with $N=|\mathcal{X}|=2^k$, $N_s=|\mathcal{X}_s|$ and $M_s=N/N_s=|\mathcal{F}_x^+|=|\mathcal{B}^{(D)}_x|$ for $x\in\mathcal{X}_s$, where $\mathcal{B}^{(D)}_x$ is a phased-in code.}
      \label{fig-sAEDS-D3}
 \end{center}
\vspace*{-0.2cm}
\end{figure}

\begin{theorem}\label{th-5}
In the case that $N$ is a power of 2, the average code length $L$ of the above sAEDS is bounded by
\begin{align}
L&\leq H(p)+D(p\|q)+\sum_{s\in\mathcal{S}}p(s)(\nu_{N_s, \tilde{Q}_s}-\mu_{\mathrm{pi}}(N_s))\nonumber\\
&\leq H(p)+D(p\|q)+\sum_{s\in\mathcal{S}}p(s)\nu_{N_s,\tilde{Q}_s}, \label{eq5-th3}
\end{align}
and $\mu_{\mathrm{pi}}(N_s)$ and $\nu_{N_s, \tilde{Q}_{s}}$ are given by
\begin{align}
0\leq \mu_{\mathrm{pi}}(N_s)& = k_s+1-(2^{k_s}/N_s)-\lg N_s\leq \sigma,\label{eq5-th3-1}\\
0\leq \nu_{N_s, \tilde{Q}_{s}}&=\frac{2N_s-2^{k_s}}{N_s}-\sum_{x\in\check{\mathcal{X}_s}}\tilde{Q}_s(x)<\frac{2N_s-2^{k_s}}{N_s}\leq 1,\label{eq5-th3-2}
\end{align}
where $k_s=\lceil \lg N_s\rceil$, $\tilde{Q}_s(x)=\sum_{\hat{x}\in\mathcal{F}_x^+}Q(\hat{x})$ for $x\in \mathcal{X}_s$, and $\check{\mathcal{X}}_s$ is the subset of $\mathcal{X}_s$ with $2N_s-2^{k_s}$ elements of smaller probabilities $\tilde{Q}_s(x)$.
\end{theorem}
The proof of Theorem \ref{th-5} is given in Appendix \ref{App-E}.
$\mu_{\mathrm{pi}}(N_s)$ is the redundancy of the phased-in code
in the case that $\tilde{Q}_s=\{\tilde{Q}_s(x) \,|\, x\in\mathcal{X}_s\}$ 
is a uniform distribution, and $\nu_{N_s, \tilde{Q}_{s}}$ represents 
a deviation of $\tilde{Q}_s$ from the uniform distribution.
Refer Appendix~\ref{App-I} for more details, where the average code length and redundancy of
the phased-in code is evaluated.

\begin{remark} 
In the sAEDSs used in Theorems \ref{th-2}-\ref{th-5},
the length difference of codewords is within 1 bit
in $\bigcup_{x\in\mathcal{X}_s} \mathcal{B}^{(D)}_x$ for each $s\in\mathcal{S}$.  
Hence, each of these sAEDSs can be realized as a tANS by
selecting functions $C[s,y]$ and $D[x]$ appropriately. Refer Remark \ref{rem-2}.
Therefore, Theorems \ref{th-2}-\ref{th-5} also hold for the optimal tANS with $N=2^k$ states. 
\end{remark}

\section{Performance of the sAEDS in large $N$}
Yokoo and Dub\'{e} \cite{YD2019} proved that the tANS with their functions $C[s,y]$ and $D[x]$ is asymptotically optimal, i.e.,~the average code length $L$ of the tANS goes to $H(p)$ when $N$ tends to infinity. Since the tANS can be considered as a special case of sAEDS, the optimal sAEDS is also asymptotically optimal.
However, they have not studied about how fast $L$ converges to $H(p)$ as $N$ increases. 
In this section, we consider  the sAEDS shown in Fig.~\ref{fig6},\footnote{
This sAEDS is equivalent to the tANS treated in \cite{YD2019} in the sense
that this sAEDS has the same codeword length as the tANS.} and will show
that this sAEDS satisfies  $L \leq H(p)+O(1/N)$.

\begin{figure}[t]
 \begin{center}
   \includegraphics[width=7cm]{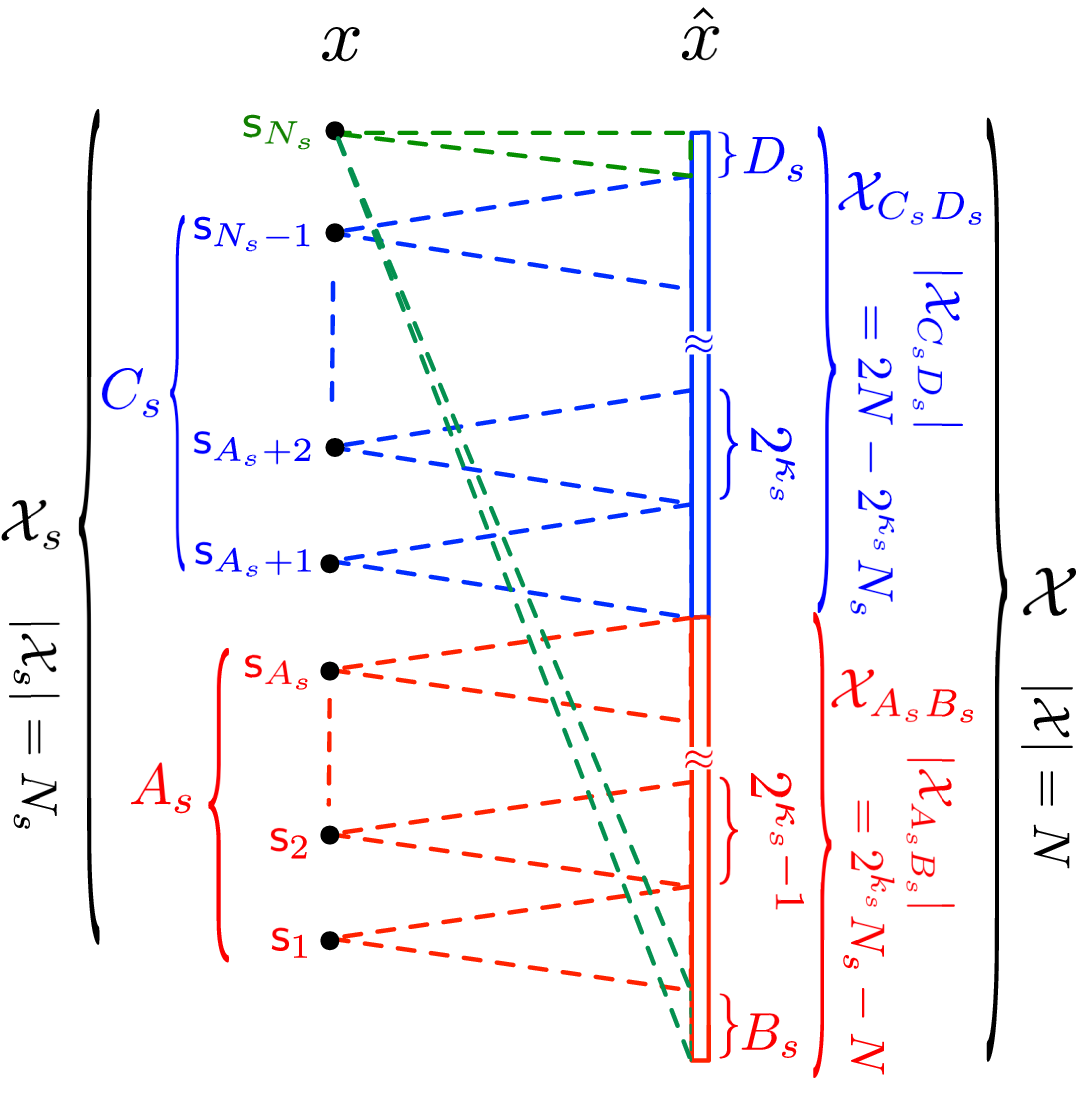}
    \caption{An sAEDS, which is equivalent to the tANS treated in \cite{YD2019}.}
      \label{fig6}
 \end{center}
\end{figure}

For each $s\in\mathcal{S}$ with $N_s=|\mathcal{X}_s|$, we define $\kappa_s$ as
\begin{align}
\kappa_s=\left\lceil \lg \frac{N}{N_s} \right\rceil. \label{eq6-0}
\end{align}
In Fig.~\ref{fig6}, the right side represents all states $\hat{x}\in\mathcal{X}$ such that states $\hat{x}$ are arranged
from bottom to top in descending order of probability $Q(\hat{x})$.
On the other hand, the left side represents $\mathcal{X}_s=\{\s_1, \s_2, \cdots, \s_{N_s}\}$ for 
a source symbol $s\in\mathcal{S}$.
Then, we divide $\mathcal{X}_s$ into three groups: $\{\s_i \,|\, 1\leq i \leq A_s\}$ with $A_s$ states, $\{\s_i \,|\, A_s+1\leq i \leq N_s-1\}$ with $C_s(=N_s-A_s-1)$ states, and $\{\s_{N_s}\}$ with one state.
For each state $\s_i$ in the first, second, and third groups, we set $|\mathcal{F}^{+}_{\s_i}|=2^{\kappa_s-1}$,
$|\mathcal{F}^{+}_{\s_i}|=2^{\kappa_s}$, and $|\mathcal{F}^{+}_{\s_{N_s}}|=B_s+D_s$, respectively. 
Note that codeword length becomes 
$\kappa_s-1$ bits and $\kappa_s$ bits for  the first and second groups, respectively, when  we use a fixed length code.

From Fig.~\ref{fig6}, parameters $A_s, B_s, C_s$, and $D_s$ must satisfy $2^{\kappa_s}N_s-N=2^{\kappa_s-1}A_s+B_s$, $2N-2^{\kappa_s}N_s=2^{\kappa_s}C_s+D_s$, and $N_s=A_s+C_s+1$. Hence, by deleting $A_s$ and $C_s$ from these
equations, we have the following relation.
\begin{align}
2B_s+D_s=2^{\kappa_s}.\label{eq6-3}
\end{align}
This relation means that for $\mathcal{B}^{(D)}_{\mathsf{s}_{N_s}}$,
we can use a phased-in code with  $\kappa_s$ bits for $D_s$ states and $\kappa_s-1$ bits for $B_s$ states in $\mathcal{F}^{+}_{\s_{N_s}}$.
Hence, when we divide $\mathcal{X}$ into $\mathcal{X}_{A_sB_s}$ with larger $Q(\hat{x})$
and $\mathcal{X}_{C_sD_s}$ with smaller $Q(\hat{x})$ 
so that $|\mathcal{X}_{A_sB_s}|=2^{\kappa_s}N_s-N$ and  $|\mathcal{X}_{C_sD_s}|=2N-2^{\kappa_s}N_s$
as shown in Fig.~\ref{fig6},
we have $l(E_{\hat{x}}(s))=\kappa_s-1$ if $\hat{x}\in\mathcal{X}_{A_sB_s}$
and $l(E_{\hat{x}}(s))=\kappa_s$ if  $\hat{x}\in \mathcal{X}_{C_sD_s}$.
We now define $\tilde{Q}_s$ by
\begin{align}
\tilde{Q}_s =\sum_{\hat{x}\in\mathcal{X}_{C_sD_s}} Q(\hat{x}).\label{eq6-4}
\end{align}
Then, the average code length is given from \eqref{eq3-4} by 
\begin{align}
L&=\sum_{s\in\mathcal{S}}p(s)\left[\sum_{\hat{x}\in\mathcal{X}_{A_sB_s}}Q(\hat{x}) (\kappa_s-1)
  +\sum_{\hat{x}\in\mathcal{X}_{C_sD_s}}Q(\hat{x}) \kappa_s\right]\nonumber\\
  &= \sum_{s\in\mathcal{S}}p(s) (\kappa_s -1+\tilde{Q}_s). \label{eq6-5}
\end{align}

For $\mathcal{X}=\{\alpha_1, \alpha_2, \cdots, \alpha_N\}$, 
we define a probability distribution $Q^*(\alpha_i)$ by
\begin{align}
Q^*(\alpha_i)=\lg\frac{N+i}{N+i-1},  \label{eq6-6-0} 
\end{align}
which satisfies
\begin{align}
\sum_{i=1}^NQ^*(\alpha_i)=\sum_{i=1}^N\lg \frac{N+i}{N+i-1}=\lg \frac{2N}{N}=1.\label{eq6-6}
\end{align}
Note that this probability distribution $\{Q^*(\alpha_i)\}$ is equivalent to 
the optimal stationary probability distribution given by Yokoo and Dub\'{e} for the tANS \cite{DY2019}\cite{YD2019}.

Let  $Q(\alpha_i)$ be the stationary probability distribution of the AEDS given in Fig.~\ref{fig6}.
Then, the following lemma holds.
\begin{lemma}\label{lm-1}
If it holds that $Q(\alpha_i)=Q^*(\alpha_i)$ for all $\alpha_i\in\mathcal{X}$, then the average code length $L$ satisfies
\begin{align}
L=H(p)+D(p\|q), \label{eqL1-1}
\end{align}
where $q=\{q(s)=N_s/N \;|\; s\in\mathcal{S}\}$.
\end{lemma}

\begin{lemma}\label{lm-2}
If $Q(\alpha_i)$ satisfies  
\begin{align}
Q(\alpha_i)&< Q^*(\alpha_i)+\frac{\eta}{N^2},  \label{eqL2-1}\\
Q(\alpha_i)&< Q^*(\alpha_i)+\frac{\eta}{N\lg N},  \label{eqL2-2}\\
Q(\alpha_i)&< Q^*(\alpha_i)+\frac{\eta}{N},\label{eqL2-3}
\end{align}
for a constant $\eta> 0$ and all $\alpha_i\in\mathcal{X}$, then the average code length $L$ is upper bounded by
\begin{align}
L&< H(p)+D(p\|q)+\frac{\eta}{N},\label{eqL2-4}\\
L&< H(p)+D(p\|q) +\frac{\eta}{\lg N},\label{eqL2-5}\\
L&< H(p)+D(p\|q) + \eta,\label{eqL2-6}
\end{align}
respectively.
\end{lemma}
The proofs of Lemmas \ref{lm-1} and \ref{lm-2} are given in Appendix \ref{App-F}.

As an example of probability distribution that satisfies \eqref{eqL2-1},
we consider $Q^o(\alpha_i)$ defined by
\begin{align}
Q^o(\alpha_i)&=\frac{\theta}{N+i-1}, \label{eqL2-6-1}
\end{align}
where $\theta>0$ is the normalization constant to satisfy $\sum_{i=1}^N Q^o(\alpha_i)=1$.
Note that $\{Q^o(\alpha_i)\}$ is equivalent to the state probability distribution treated in \cite[Section V]{PDPCMM} and \cite[Section 6]{PDPCMM2023} for the tANS.
As the second example, we consider $\{Q^*_{\gamma}(\alpha_i)\}$ defined by
\begin{align}
Q^*_{\gamma}(\alpha_i)&=\lg\frac{N+[i-\gamma]_1}{N+[i-\gamma]_1-1}, \label{eqL2-6-2}
\end{align}
where $[a]_1=\max\{a,1\}$ and $\gamma>2$ is a constant.

Then the following Lemmas hold.
\begin{lemma}\label{lm-3}
$Q^o(\alpha_i)$  satisfies
\begin{align}
Q^o(\alpha_i)&<Q^*(\alpha_i)+\frac{\lg \text{e}}{2N^2}. \label{eqL3-1}
\end{align}
If it holds that $Q(\alpha_i)= Q^o(\alpha_i)$ for all $\alpha_i\in\mathcal{X}$, 
the average code length $L$ of the sAEDS satisfies
\begin{align}
L < H(S)+D(p\|q)+\frac{\lg e}{2N}. \label{eqL3-2}
\end{align}
\end{lemma}

\begin{lemma}\label{lm-4}
$Q^*_{\gamma}(\alpha_i)$ satisfies
\begin{align}
Q^*_{\gamma}(\alpha_i)&<Q^*(\alpha_i)+ \frac{(\gamma+1/2)\lg e}{N^2}, \label{eqL4-1}\\
Q^*_{\gamma}(\alpha_i)&>Q^*(\alpha_i)+\frac{(\gamma-2)\lg e}{4N^2}. \label{eqL4-2}
\end{align}
If it holds that $Q(\alpha_i) \leq Q^*_{\gamma}(\alpha_i)$ for all $\alpha_i\in\mathcal{X}$, the average code length $L$ of the sAEDS satisfies
\begin{align}
L < H(S)+D(p\|q)+\frac{(\gamma+1/2)\lg e}{N}. \label{eqL4-3}
\end{align}
\end{lemma}
The proofs of Lemmas \ref{lm-3} and \ref{lm-4} are given in Appendix \ref{App-G}.

We note that since $\lg\frac{N+i}{N+i-1}$ is a monotonically decreasing function of $i$,
we have $\lg\frac{N+i}{N+i-1}< \lg\frac{N+[i-\gamma]_+}{N+[i-\gamma]_+-1}$.
In the case of $Q(\alpha_i)=Q^*(\alpha_i)$,  relation $Q(\alpha_{i-1})>Q(\alpha_i)$ must hold for every $i$.
However, $Q(\alpha_i)$ satisfying $Q(\alpha_i)\leq Q^*_{\gamma}(\alpha_i)$ has some degree of freedom,
and it is allowed that multiple states $\alpha_i$ takes the same stationary probability\footnote{When $p(s)=p(s')$ for $s\neq s'$, we guess that
the optimal sAEDS has the same code structure for $\mathcal{X}_s$ and $\mathcal{X}_{s'}$ and hence multiple $Q(\alpha_i)$ take the same value in such cases.
}.

From \eqref{eqL4-1} and \eqref{eqL4-2}, we have  
\begin{align}
Q^*_{\gamma}(\alpha_i) -Q^*(\alpha_i)=\Theta\left(\frac{1}{N^2}\right)>0. \label{eq6-22}
\end{align}
Hence, $Q^*_{\gamma}(\alpha_i)$ is not a probability distribution because of $\sum_{i=1}^NQ^*_{\gamma}(\alpha_i)=\sum_{i=1}^N[Q^*(\alpha_i)+\Theta(1/N^2)]=1+\Theta(1/N)>1$. This fact 
enables that $Q(\alpha_i)$ satisfies  $Q(\alpha_i) \leq Q^*_{\gamma}(\alpha_i)$ for all $i$
as shown in Theorem \ref{th-6}.

When $N$ is sufficiently large, we can satisfy $q(s)=N_s/N\approx p(s)$, and hence $D(p\|q)\approx 0$,
by selecting $N_s$ and $N$ appropriately. Hence, for simplicity, we consider the case of $q(s)=N_s/N=p(s)$.
Then the following theorem holds.

\begin{theorem} \label{th-6}
When $q(s)=N_s/N=p(s)$, the stationary probability $Q(\alpha_i)$ of the sAEDS given in
Fig.~\ref{fig6} satisfies $Q(\alpha_i)\leq Q^*_{\gamma}(\alpha_i)$ for all $\alpha_i\in\mathcal{X}$ 
and a certain constant $\gamma>0$,
and  the average code length $L$ of the sAEDS is upper bounded by
\begin{align}
L \leq H(p)+O\left(\frac{1}{N}\right). \label{eq6-23}
\end{align}
\end{theorem}
The proof of Theorem \ref{th-6} is given in Appendix \ref{App-H}.
Since the sAEDS of Fig.~\ref{fig6}  is equivalent to the tANS treated in \cite{YD2019}, the bound \eqref{eq6-23}
also holds for the optimal tANS.

\section{Concluding Remarks}
We proposed the AEDS in this paper. Although the AEDS can be considered as a generalization of tANS,
the code class of the AEDS is much broader than that of the tANS.
We showed that Type-I and Type-{\rtwo} AEDSs can  be easily constructed based on a code tree, and
if the half subtree of the Huffman code has a probability weight larger than 0.61803 (resp.~0.56984),
Type-I AEDS with 2 states (resp.~Type-{\rtwo} AEDS with 5 states) based on the Huffman code tree can attain an average code length shorter than the Huffman code. It is worth noting that these kinds of AEDS cannot be realized as a tANS.
We also derived several upper bounds of average code length $L$ for the sAEDS, and we showed that
the optimal sAEDS with $N$ states can attain $L\leq H(P)+O(1/N)$.
Although these bounds also hold for the optimal tANS, deriving them within the framework of tANS seems complicated.

As shown for Type-I and Type-{\rtwo} AEDSs, an AEDS with a few $N$ states can attain a better compression than the Huffman code. But, Type-I and Type-{\rtwo} AEDSs may not be optimal in all AEDSs with a given $N$.
It is an interesting future work to derive the optimal AEDS that attains the minimal average code length
for a given $N$.

\vspace{0.5cm}
\appendix

\subsection{Average code length and redundancy of phased-in code}\label{App-I}
In this appendix, we evaluate the average code length of the phased-in code \cite{phasedin}\cite{phasedin-2} for a probability distribution $Q=\{Q(\hat{x}) \;|\; \hat{x}\in\mathcal{X}, |\mathcal{X}|=M\}$.
We represent the code length of $\hat{x}\in\mathcal{X}$ by $l_{\text{pi}}(\hat{x})$
and the average code length by $L_{\text{pi}}(M|Q)=\sum_{\hat{x}\in\mathcal{X}} Q(\hat{x})l_{\text{pi}}(\hat{x})$.
Note that for $k=\lceil\lg M\rceil$, the phased-in code has $2^k-M$ codewords with $l_{\text{pi}}(\hat{x})=k-1$  and $2M-2^k$ codewords with $l_{\text{pi}}(\hat{x})=k$.

First we consider the case that $Q$ is a uniform distribution $Q_{u}=\{Q(\hat{x})=1/M\; \text{for all $\hat{x}\in\mathcal{X}$}\}$. In this case, the redundancy of the phased-in code is defined by $\mu_{\text{pi}}(M)=L_{\text{pi}}(M|Q_u)-\lg M$.
Then, the following lemma holds.\footnote{From \cite[Theorem 2]{Gallager}, we know $\mu_{\text{pi}}(M)\leq (1/M)+\sigma$. But, \eqref{eq-A4} is tighter than this bound.}
\begin{lemma} \label{lemma-pi-u}
When $Q$ is a  uniform distribution $Q_{u}=\{Q(\hat{x})=1/M\; \text{for all $\hat{x}\in\mathcal{X}$}\}$, the average code length $L_{\mathrm{pi}}(M|Q_u)$ and
the redundancy $\mu_{\mathrm{pi}}(M)$ of the phased-in code satisfy
\begin{align}
\lg M \leq & L_{\mathrm{pi}}(M|Q_u) \leq \lg M+\sigma, \label{eq-A5}\\
0\leq & \mu_{\mathrm{pi}}(M)\leq \sigma,  \label{eq-A4}
\end{align}
where $\sigma=\lg\lg e + 1 -\lg e\approx 0.08607$.
\end{lemma}
\begin{proof}
For the uniform distribution $Q_u$, $L_{\text{pi}}(M|Q_u)$ and $\mu_{\text{pi}}(M)$ are given by
\begin{align}
L_{\text{pi}}(M|Q_u)&=\sum_{\hat{x}\in\mathcal{X}} \frac{1}{M}l_{\text{pi}}(\hat{x}) \nonumber\\
&= \frac{1}{M}\left[(2^k-M)(k-1)+(2M-2^k)k\right] \nonumber\\
&=k+1-\frac{2^k}{M}, \label{eq-A1}\\
\mu_{\text{pi}}(M)&=L_{\text{pi}}(M|Q_u)-\lg M \label{eq-A1-1}\\
&=k+1-\frac{2^k}{M}-\lg M. \label{eq-A2}
\end{align}
Defining $a=k-\lg M$,  $0\leq a <1$,  $\mu_{\text{pi}}(M)$ can be
represented by
\begin{align}
\mu_{\text{pi}}(M)=a +1 -\frac{2^{a+\lg M}}{M}=a+1-2^a.   \label{eq-A3}
\end{align}
Then, we can easily check that $\mu_{\text{pi}}(M)$
is maximized at $a=\lg\lg e$ and minimized at $a=0$ and $a\to1$. 
Therefore, $\mu_{\text{pi}}(M)$ satisfies \eqref{eq-A4},
and  \eqref{eq-A5} is obtained from \eqref{eq-A4} and \eqref{eq-A1-1}.
\end{proof}

Next we consider a general probability distribution $Q=\{Q(\hat{x}) \;|\; \hat{x}\in\mathcal{X}\}$.
We divide $\mathcal{X}$ into two subsets, $\hat{\mathcal{X}}$ with larger probabilities $Q(\hat{x})$ and $\check{\mathcal{X}}$ with smaller probabilities $Q(\hat{x})$ such that $|\hat{\mathcal{X}}|=2^k-M$ and $|\check{\mathcal{X}}|=2M-2^k$. We define $\hat{Q}$, $\check{Q}$, and $\nu_{M,Q}$ as follows.
\begin{align}
\hat{Q}&=Q(\hat{\mathcal{X}})=\sum_{\hat{x}\in\hat{\mathcal{X}}}Q(\hat{x}), \label{eq-A6}\\
\check{Q}&=Q(\check{\mathcal{X}})=\sum_{\hat{x}\in\check{\mathcal{X}}}Q(\hat{x}), \label{eq-A7}\\
\nu_{M,Q}&=\hat{Q}-\frac{2^k-M}{M}=\frac{2M-2^k}{M}-\check{Q}. \label{eq-A8}
\end{align}
Note that $\nu_{M,Q}$ represents the deviation of $Q$ from  the uniform distribution, and it satisfies  
\begin{align}
0\leq \nu_{M,Q}< \frac{2M-2^k}{M}\leq 1.  \label{eq-A9}
\end{align}

Then, the following lemma holds.
\begin{lemma} \label{lemma-pi-u}
The average code length $L_{\mathrm{pi}}(M|Q)$ of the phased-in code
satisfies
\begin{align}
L_{\mathrm{pi}}(M|Q)&=\lg M +\mu_{\mathrm{pi}}(M)-\nu_{M,Q}\label{eq-A10}\\
&\leq \lg M+\sigma-\nu_{M,Q}. \label{eq-A12}
\end{align}
\end{lemma}
\begin{proof}
In the phased-in code,
$\hat{x}\in\hat{\mathcal{X}}$ is encoded into a $(k-1)$-bit codeword and $\hat{x}\in\check{\mathcal{X}}$ is encoded into a $k$-bit codeword.
Hence, the average code length $L_{\text{pi}}(M|Q)$ is given by
\begin{align}
L_{\mathrm{pi}}(M|Q)&=k-\hat{Q}\nonumber\\
&=^{*1} k-\left(\frac{2^k-M}{M}+\nu_{M,Q}\right) \nonumber\\
&=k+1-\frac{2^k}{M} -\nu_{M,Q} \nonumber\\
&=^{*2} \lg M +\mu_{\text{pi}}(M)-\nu_{M,Q} \nonumber\\
&\leq^{*3} \lg M +\sigma -\nu_{M,Q},
\end{align}
where $=^{*1}$, $=^{*2}$, and $\leq^{*3}$ are obtained from \eqref{eq-A8},
and \eqref{eq-A2} and \eqref{eq-A4}, respectively. 
\end{proof}

\subsection{Proof of Theorem \ref{theorem-I}}\label{App-A}

From the state transitions shown in Fig.~\ref{fig-I1}, the stationary probabilities $Q(\alpha_j)$ satisfy
\begin{align}
\sum_{j=1}^N Q(\alpha_j)&=1, \label{eq-I2-1}\\
Q(\alpha_{j+1})&=Q(\alpha_j) P_R\; \text{ for $1\leq j \leq N-1$}, \label{eq-I3}
\end{align}
and hence we have
\begin{align}
Q(\alpha_j)=\frac{P_R^{j-1}(1-P_R)}{1-P_R^N}. \label{eq-I4}
\end{align}
On the other hand,   the average code lengths 
$\sum_{s\in\mathcal{S}_R}  p(s)l(E_{\hat{x}}(s))$ and $\sum_{s\in\mathcal{S}_L}  p(s)l(E_{\hat{x}}(s))$
for $\hat{x}=\alpha_j$ are 
given in Table \ref{table-I1}.
\begin{table}[t]
\begin{center}
\caption{the average code lengths 
$\sum_{s\in\mathcal{S}_R}  p(s)l(E_{\hat{x}}(s))$ and $\sum_{s\in\mathcal{S}_L}  p(s)l(E_{\hat{x}}(s))$ for $\hat{x}=\alpha_j$}
\label{table-I1}
\begin{tabular}{r||c|c} 
 & for $\mathcal{S}_R$ & for $\mathcal{S}_L$\\ \hline
$1\leq  j \leq 2^k-N$  & $L_{T_R}$& $L_{T_L}+k P_L$ \\
$2^k-N+1 \leq j \leq N-1$\hspace{0.18cm}  & $L_{T_R}$ & $L_{T_L}+(k+1) P_L$\\
$j=N$\hspace{0.7cm}  & $L_{T_R}+P_R$ &$L_{T_L}+(k+1) P_L$
\end{tabular}
\end{center}
\end{table}
Therefore, from \eqref{eq3-3} and Table \ref{table-I1}, 
$L_N^\I$  is obtained as follows.
\begin{align}
L_N^\I&=\sum_{\hat{x}\in\mathcal{X}} Q(\hat{x})\sum_{s\in\mathcal{S}}  p(s)l(E_{\hat{x}}(s))\nonumber\\
&=\sum_{j=1}^{2^k-N} Q(\alpha_j)[L_{T_R}+L_{T_L}+kP_L]
+\sum_{j=2^k-N+1}^{N-1}Q(\alpha_j)[L_{T_R}+L_{T_L}+(k+1)P_L] \nonumber\\
&\hspace{5.15cm}+Q(\alpha_N)[L_{T_R}+P_R +L_{T_L}+(k+1)P_L]\nonumber\\
&=^{*1}\sum_{j=1}^{2^k-N} Q(\alpha_j)[L_T-P_R+(k-1)P_L] 
+\sum_{j=2^k-N+1}^{N-1} Q(\alpha_j)[L_T-P_R+kP_L]\nonumber\\
&\hspace{6.05cm}+Q(\alpha_N)[L_T+kP_L]\nonumber\\
&=^{*2} L_T -(1-Q(\alpha_N))P_R +kP_L-\sum_{j=1}^{2^k-N}Q(\alpha_j)P_L\nonumber\\
&=^{*3} L_T-\left(1-\frac{P_R^{N-1}(1-P_R)}{1-P_R^N}\right)P_R +k(1-P_R)
-\frac{1-P_R^{2^k-N}}{1-P_R^N}(1-P_R)\nonumber\\
&= L_T-\left[\frac{1-P_R^{N-1}}{1-P_R^N}P_R +\frac{1-P_R^{2^k-N}}{1-P_R^N}(1-P_R)-k(1-P_R)\right],
\label{eq-I5}
\end{align}
where $=^{*1}$, $=^{*2}$, and $=^{*3}$ hold from \eqref{eqT-1}, \eqref{eq-I2-1}, and \{\eqref{eq-I4} and $P_L=1-P_R$\}, respectively.

Hence, the reduction $\delta_N^\I(P_R)$ in \eqref{eq-I1} is given by
\begin{align}
\delta_N^\I(P_R)=&[L_T-L_N^\I]_0 \nonumber\\
=&\left[\frac{1-P_R^{N-1}}{1-P_R^N}P_R +\frac{1-P_R^{2^k-N}}{1-P_R^N}(1-P_R)-k(1-P_R)\right]_0.
\label{eq-I6}
\end{align}
When $N=2$, $\delta_2^\I(P_R)$ is given by
\begin{align}
\delta_2^\I(P_R)&=\left[\frac{P_R^2+P_R-1}{1+P_R}\right]_0.\label{eq-I7}
\end{align}
Then, we can easily check that $\delta_2^\I(P_R)$ is positive for $P_R>\omega^\I$ where $\omega^\I=(-1+\sqrt{5})/2\approx 0.6180$.

\subsection{Proof of Theorem \ref{theorem-II}}\label{App-B}

From the state transitions of Type-\rtwo~AEDS with 5 states shown in Fig.~\ref{fig-II1}, the stationary probabilities $Q(\alpha_j)$ satisfies
the following relations.
\begin{align}
\sum_{j=1}^5 Q(\alpha_j)&=1, \label{eq-II3}\\
Q(\alpha_1)&=(Q(\alpha_2)+Q(\alpha_3)+Q(\alpha_4)+Q(\alpha_5))\, P_L, \label{eq-II4}\\
Q(\alpha_2)&= Q(\alpha_1)\,P_L, \label{eq-II5}\\
Q(\alpha_3)&=(Q(\alpha_1)+Q(\alpha_2)+Q(\alpha_5))\, P_R, \label{eq-II6}\\
Q(\alpha_4)&=Q(\alpha_3)\,P_R, \label{eq-II7}\\
Q(\alpha_5)&=Q(\alpha_4)\,P_R. \label{eq-II8}
\end{align}
From \eqref{eq-II3}--\eqref{eq-II8}, we obtain
\begin{align}
Q(\alpha_1)&=\frac{P_L}{1+P_L}=\frac{1-P_R}{2-P_R}, \label{eq-II9}\\
Q(\alpha_2)&=\frac{P_L^2}{1+P_L} =\frac{(1-P_R)^2}{2-P_R}, \label{eq-II10}\\
Q(\alpha_3)&=\frac{P_R}{1+P_R+P_R^2}, \label{eq-II11}\\
Q(\alpha_4)&=\frac{P_R^2}{1+P_R+P_R^2}, \label{eq-II12}\\
Q(\alpha_5)&=\frac{P_R^3}{1+P_R+P_R^2}. \label{eq-II13}
\end{align}
From \eqref{eq3-3} and Fig~\ref{fig-II1}, the average code length $L^{\II}$ of Type-\rtwo~AEDS is given by
\begin{align}
L^{\II}&
=\sum_{\hat{x}\in\mathcal{X}} Q(\hat{x})\sum_{s\in\mathcal{S}}  p(s)l(E_{\hat{x}}(s))\nonumber\\
&=Q(\alpha_1)(L_{T_R}+P_R+L_{T_L})
+Q(\alpha_2)(L_{T_R}+2P_R+L_{T_L}+3P_L)+Q(\alpha_3)(L_{T_R}+L_{T_L}+P_L)\nonumber\\
&\hspace*{1cm}+Q(\alpha_4)(L_{T_R}+L_{T_L}+2P_L)+Q(\alpha_5)(L_{T_R}+2P_R+L_{T_L}+3P_L)\nonumber\\
&=^{*1} Q(\alpha_1)(L_T-P_L)+Q(\alpha_2)(L_T+P_R+2P_L)+Q(\alpha_3)(L_T-P_R)\nonumber\\
&\hspace*{1cm}+Q(\alpha_4)(L_{T}-P_R+P_L)+Q(\alpha_5)(L_T+P_R+2P_L)\nonumber\\
&=^{*2} L_T - Q(\alpha_1)(1-P_R)+Q(\alpha_2)(2-P_R)-Q(\alpha_3)P_R+Q(\alpha_4)(1-2P_R)+Q(\alpha_5)(2-P_R)\nonumber\\
&=^{*3} L_T -\frac{P_R^3-P_R^2+2P_R-1}{(2-P_R)(1+P_R+P_R^2)}, \label{eq-II14}
\end{align}
where$=^{*1}$, $=^{*2}$, and $=^{*3}$ hold because of \eqref{eqT-1}, \{\eqref{eq-II3} and $P_L=1-P_R$\}, and \eqref{eq-II9}--\eqref{eq-II13}, respectively.

Hence, $\delta^{\II}(P_R)$ in \eqref{eq-II1} is given by
\begin{align}
\delta^{\II}(P_R)&=[L^{\II}-L_T]_0 \nonumber\\
&=\left[\frac{P_R^3-P_R^2+2P_R-1}{(2-P_R)(1+P_R+P_R^2)}\right]_0. \label{eq-I15}
\end{align}
We can easily check that $\delta^{\II}(P_R)$ is positive for $P_R> \omega^{\II}\approx 0.56984$.

\subsection{Proof of Theorem \ref{th-3}}\label{App-C}

For $x\in\mathcal{X}_s$ and $\hat{x}\in\mathcal{F}^+_x$, we define $\tilde{Q}_s(x)=\sum_{\hat{x}\in\mathcal{F}_x^{+}}Q(\hat{x})$
and $\tilde{Q}_{s,x}(\hat{x})=Q(\hat{x})/\tilde{Q}_s(x)$. Then, $\{\tilde{Q}_{s,x}(\hat{x})\}$ becomes a probability distribution over $\mathcal{F}^{+}_x$.
Letting $l_{\text{pi}}(\hat{x})= l(E_{\hat{x}}(s))$ be the code length of the phased-in code and 
$L_{\text{pi}}(M_s|\tilde{Q}_{s,x})=\sum_{\hat{x}\in\mathcal{F}_x^{+}}l_{\text{pi}}(\hat{x})\tilde{Q}_{s,x}(\hat{x})$ 
be the average code length of the phased-in code over $\mathcal{F}_x^{+}$,
the average code length of the sAEDS is given from \eqref{eq3-4} as follows.
\begin{align}
L&=\sum_{s\in\mathcal{S}}p(s)\sum_{x\in\mathcal{X}_s}\sum_{\hat{x}\in\mathcal{F}_x^{+}} l_{\text{pi}}(\hat{x})Q(\hat{x})\nonumber\\
&=\sum_{s\in\mathcal{S}}p(s)\sum_{x\in\mathcal{X}_s} \tilde{Q}_s(x)\sum_{\hat{x}\in\mathcal{F}_x^{+}}l_{\text{pi}}(\hat{x})\tilde{Q}_{s,x}(\hat{x}) \nonumber\\
&=\sum_{s\in\mathcal{S}}p(s)\sum_{x\in\mathcal{X}_s} \tilde{Q}_s(x) L_{\text{pi}}(M_s|\tilde{Q}_{s,x}).\label{eq5-0}
\end{align}
Furthermore, using \eqref{eq-A9} and \eqref{eq-A12} in Appendix~\ref{App-I},  
$L$ can be bounded by
\begin{align}
L&\leq \sum_{s\in\mathcal{S}}p(s)\sum_{x\in\mathcal{X}_s} \tilde{Q}_s(x) 
(\lg M_s+\sigma) \label{eq5-1}\\
&=  \sum_{s\in\mathcal{S}}p(s) (\lg M_s +\sigma)  \nonumber\\
&= H(p) +D(p\|q) +\sigma, \label{eq5-3}
\end{align}
where the last equality holds because of $M_s=1/q(s)$.

\subsection{Proof of Theorem \ref{th-4}}\label{App-D}

In the same way as \eqref{eq5-1},  the average code length $L$ is 
bounded as follows.
\begin{align}
L &\leq \sum_{s\in\mathcal{S}}p(s)\left[(1-\tilde{Q}_{M_s+1}) (\lg M_s+\sigma)
+\tilde{Q}_{M_s+1}(\lg (M_s+1)+\sigma)\right]\nonumber\\
&\leq \sum_{s\in\mathcal{S}}p(s)\left[\lg (M_s+\tilde{Q}_{M_s+1})\right]+\sigma \nonumber\\
&=\sum_{s\in\mathcal{S}}p(s)\left[\lg \frac{N}{N_s}+ \lg \frac{M_s+\tilde{Q}_{M_s+1}}{N/N_s}\right]
+\sigma\nonumber\\
&=H(p) + D(p\|q) +\sigma+ \sum_{s\in\mathcal{S}}p(s)\lg \frac{M_s+\tilde{Q}_{M_s+1}}{N/N_s},
\label{eq5-6}
\end{align}
where the second inequality holds from Jensen's inequality of $\lg$ function,
and the final equality follows from $q(s)=N_s/N$.

\subsection{Proof of Theorem \ref{th-5}}\label{App-E}

The fixed length codes used in the AEDS of Fig.~\ref{fig-sAEDS-D3} can be constructed as follows. 
For $N=2^k$, we first construct a code tree $\mathsf{T}$ of a fixed length code with $k$ bits and the code tree $\mathsf{T}_{\text{pi}}$ of a phased-in code with $N_s$ leaves.
Then, by removing $\mathsf{T}_{\text{pi}}$ from the root of $\mathsf{T}$ as shown in Fig.~\ref{fig4}, 
we obtain $N_s$ code trees, each of which can be used as the code tree of  $\mathcal{B}^{(D)}_x$ corresponding to $\mathcal{F}^+_x$, $x\in\mathcal{X}_s$. 
\begin{figure}[t]
 \begin{center}
   \includegraphics[width=7cm]{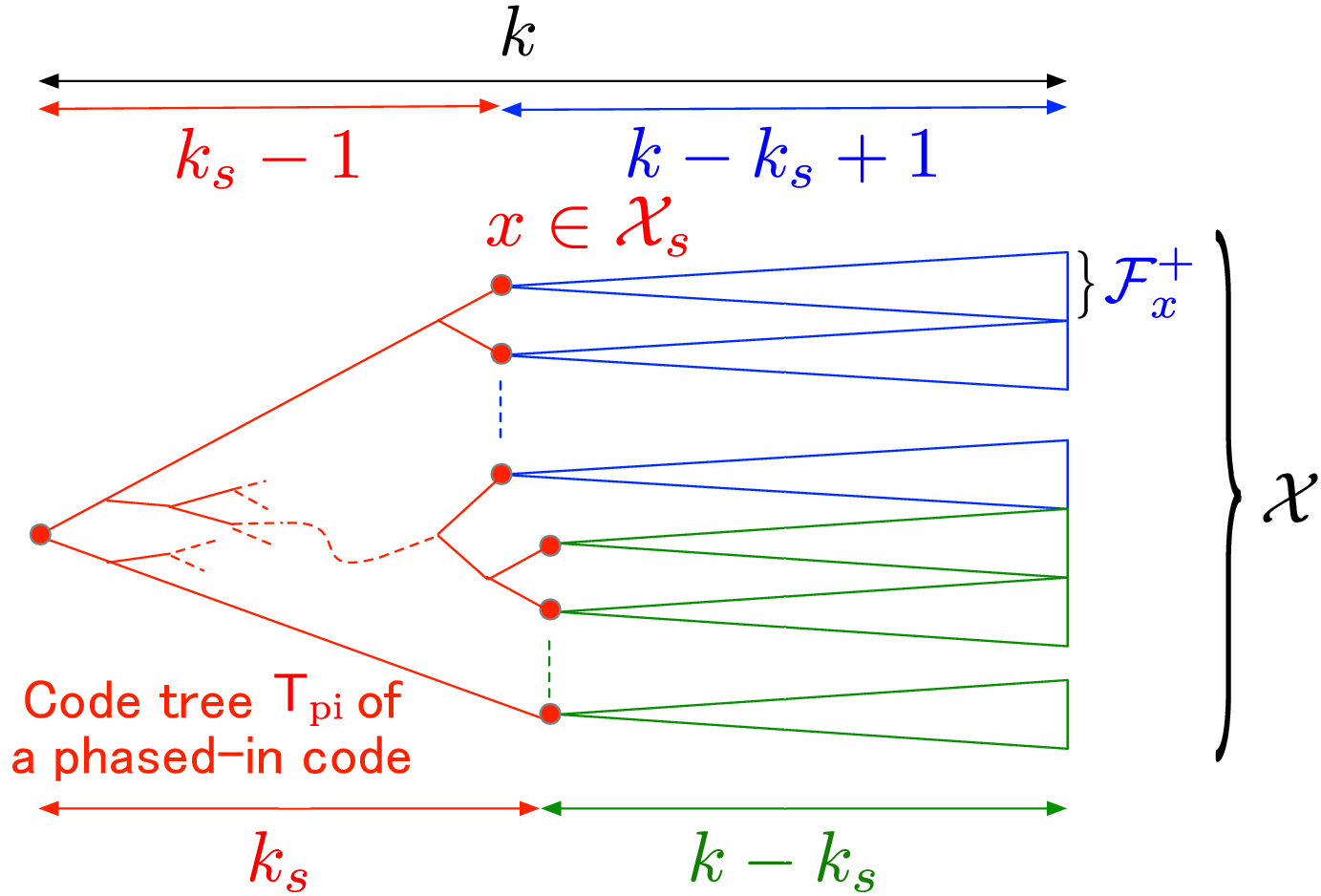}
    \caption{Code trees of $\mathcal{B}_x^{(D)}$ for $x\in\mathcal{X}_s$ in the case of $N=2^k$.}
      \label{fig4}
 \end{center}
\vspace*{-0.2cm}
\end{figure}

The average code length of the above constructed sAEDS is given from \eqref{eq3-4} by
\begin{align}
L&=\sum_{s\in\mathcal{S}}p(s)\sum_{x\in\mathcal{X}_s}\sum_{\hat{x}\in\mathcal{F}_x^{+}}
(k-l_{\text{pi}}(x))Q(\hat{x})\nonumber\\
&=\sum_{s\in\mathcal{S}}p(s)\left(k-\sum_{x\in\mathcal{X}_s}l_{\text{pi}}(x)\tilde{Q}_{s}(x)\right),\label{eq5-10-0}
\end{align}
where $\tilde{Q}_{s}(x)=\sum_{\hat{x}\in\mathcal{F}_x^+} Q(\hat{x})$ for $x\in\mathcal{X}_s$.
Since $l_{\text{pi}}(x)=k_s$ or $k_s-1$, we should assign $x\in \mathcal{X}_s$ 
with larger probabilities $\tilde{Q}_{s}(x)$ to a codeword with $k_s$ bits to minimize $L$ by maximizing $\sum_{x\in\mathcal{X}_s}l_{\text{pi}}(x)\tilde{Q}_{s}(x)$. 
However, if we control the assignment of $\mathcal{F}^+_x$ based on $\tilde{Q}_{s}(x)$,
the stationary probability distribution $\{Q(\hat{x}) |\hat{x}\in\mathcal{X}\}$ will change. 
Hence, we evaluate \eqref{eq5-10-0}  without changing the stationary probability distribution by the worst case, i.e.,~by the average code length $L_{\text{pi}}(N_s|\tilde{Q}_s)$ of the phased-in code, which gives the minimum of $\sum_{x\in\mathcal{X}_s}l_{\text{pi}}(x)\tilde{Q}_{s}(x)$ treated in Appendix \ref{App-I}.

From \eqref{eq5-10-0}, \eqref{eq-A10} in Appendix~\ref{App-I}, and $q(s)=N_s/N$,  we obtain
\begin{align}
L& \leq \sum_{s\in\mathcal{S}}p(s)\left(k-L_{\text{pi}}(N_s|\tilde{Q}_s)\right)\nonumber\\
&= \sum_{s\in\mathcal{S}}p(s)\left(\lg N - (\lg N_s+\mu_{\text{pi}}(N_s)-\nu_{N_s, \tilde{Q}_s})\right)\nonumber\\
&=H(p)+D(p\|q)+\sum_{s\in\mathcal{S}}p(s)(\nu_{N_s, \tilde{Q}_s}-\mu_{\text{pi}}(N_s))\nonumber\\
&\leq H(p)+D(p\|q)+\sum_{s\in\mathcal{S}}p(s)\nu_{N_s,\tilde{Q}_s}. \label{eq5-11}
\end{align}

\subsection{Proofs of Lemmas \ref{lm-1} and \ref{lm-2}}\label{App-F}

From \eqref{eq6-4} and Fig.~\ref{fig6}, $\tilde{Q}_s^*$ of $\{Q^*(\alpha_i)\}$ is given by
\begin{align}
\tilde{Q}_s^*&=\sum_{i=2^{\kappa_s}N_s-N+1}^NQ^*(\alpha_i)=\sum_{i=2^{\kappa_s}N_s-N+1}^N\lg\frac{N+i}{N+i-1}=\lg\frac{2N}{2^{\kappa_s}N_s}=1-\kappa_s+\lg \frac{N}{N_s}. \label{eq6-7}
\end{align}
Then, from \eqref{eq6-5}, \eqref{eq6-7}, and $q(s)=N_s/N$,  \eqref{eqL1-1} in Lemma \ref{lm-1} is obtained by
\begin{align}
L&=\sum_{s\in\mathcal{S}}p(s)\lg\frac{N}{N_s} \nonumber\\
&=\sum_{s\in\mathcal{S}} p(s)\lg \frac{1}{q(s)}  \nonumber\\
&=H(p)+D(p\|q). \label{eq6-8}
\end{align}

To prove Lemma \ref{lm-2}, we first consider the case  that $Q(\alpha_i)$ satisfies \eqref{eqL2-1}.
In this case, $\tilde{Q}_s$ and $\tilde{Q}_s^*$ defined by \eqref{eq6-4} and \eqref{eq6-7} 
satisfy 
\begin{align}
\tilde{Q}_s-\tilde{Q}_s^*
&=\sum_{i=2^{\kappa_s}N_s-N+1}^N (Q(\alpha_i)-Q^*(\alpha_i))\nonumber\\
& \leq \sum_{i=2^{\kappa_s}N_s-N+1}^N \frac{\eta}{N^2}\nonumber\\
&=(2N-2^{\kappa_s}N_s)\frac{\eta}{N^2}\nonumber\\
& \leq \frac{\eta N}{N^2} \nonumber\\
&=\frac{\eta}{N}.  \label{eq6-10}
\end{align}
Hence, \eqref{eqL2-4} in Lemma \ref{lm-2}  is derived from \eqref{eq6-5}, \eqref{eq6-7}, and \eqref{eq6-10}
as follows.
\begin{align}
L&= \sum_{s\in\mathcal{S}}p(s) (\kappa_s -1+\tilde{Q}^*_s + (\tilde{Q}_s-\tilde{Q}^*_s))\nonumber\\
&\leq  \sum_{s\in\mathcal{S}}p(s) \left(\lg\frac{N}{N_s}+\frac{\eta}{N}\right) \nonumber \\
&= H(p)+D(p\|q) +\frac{\eta}{N}.\label{eq6-11}
\end{align}

In the same way, if $Q(\alpha_i)$ satisfies \eqref{eqL2-2} and \eqref{eqL2-3},
we can derive \eqref{eqL2-5} and \eqref{eqL2-6}, respectively.

\subsection{Proofs of Lemmas \ref{lm-3} and \ref{lm-4}}\label{App-G}
To prove Lemma \ref{lm-3}, we first consider the probability distribution $\{Q^o(\alpha_i)\}$ defined by \eqref{eqL2-6-1}.
Then, $Q^o(\alpha_i)$ satisfies 
\begin{align}
Q^o(\alpha_i)-Q^*(\alpha_i)
&=\frac{\theta}{N+i-1}-\lg \frac{N+i}{N+i-1}\nonumber\\
&<^{*1}\frac{\lg \text{e}}{N+i-1}-\lg \left(1+ \frac{1}{N+i-1}\right)\nonumber\\
&<^{*2} \frac{\lg \text{e}}{2(N+i-1)^2}\nonumber\\
&\leq \frac{\lg \text{e}}{2N^2}, \label{eq6-14}
\end{align}
where $<^{*1}$ and $<^{*2}$ are obtained from inequalities \eqref{eq6-15} and \eqref{eq6-16}, respectively.
\begin{align}
1&=\sum_{i=1}^{N} Q^o(\alpha_i)=\sum_{i=1}^N\frac{\theta}{N+i-1}=\sum_{i=N}^{2N-1}\frac{\theta}{i}
> \theta\int_N^{2N}\frac{1}{u}\text{d}u=\theta \ln 2=\frac{\theta}{\lg \text{e}}.\label{eq6-15}
\end{align}
\begin{align}
(\lg \text{e})\left(\frac{1}{u}-\frac{1}{2u^2}\right) < \lg \left(1+\frac{1}{u}\right)<
(\lg \text{e}) \frac{1}{u} \quad\text{for $ u>1$}. \label{eq6-16}
\end{align}
Hence, in the case of $Q(\alpha_i)= Q^o(\alpha_i)$ for all $i$, the average codeword length $L$ 
satisfies \eqref{eqL3-2} from Lemma \ref{lm-2} and \eqref{eq6-14}.

To prove Lemma \ref{lm-4}, we next consider  $\{Q^*_{\gamma}(\alpha_i)\}$ defined by
\eqref{eqL2-6-2}.
Then, $Q^*_{\gamma}(\alpha_i)$ satisfies
\begin{align}
 Q^*_{\gamma}(\alpha_i)-Q^*(\alpha_i) 
&=\lg\frac{N+[i-\gamma]_+}{N+[i-\gamma]_+-1} -\lg \frac{N+i}{N+i-1}\nonumber\\
& = \lg\left(1+\frac{1}{N+[i-\gamma]_+-1}\right) -\lg\left(1+\frac{1}{N+i-1}\right)\nonumber\\
&<^{*2}(\lg \text{e})\left(\frac{1}{N+[i-\gamma]_+-1}-\frac{1}{N+i-1}
+\frac{1}{2(N+i-1)^2}\right)\nonumber\\
&= (\lg \text{e})\left(\frac{i-[i-\gamma]_+}{(N+[i-\gamma]_+-1)(N+i-1)}+\frac{1}{2(N+i-1)^2}\right)\nonumber\\
&\leq (\lg \text{e})\left(\frac{\gamma}{(N+[i-\gamma]_+-1)(N+i-1)}+\frac{1}{2(N+i-1)^2}\right)\nonumber\\
&<\frac{(\gamma+1/2)\lg \text{e}}{N^2}. \label{eq6-19}
\end{align}

The lower bound of $Q^*_{\gamma}(\alpha_i)-Q^*(\alpha_i)$ can also be derived as follows.
\begin{align}
Q^*_{\gamma}(\alpha_i)-Q^*(\alpha_i)
&= \lg\frac{N+[i-\gamma]_+}{N+[i-\gamma]_+-1}-\lg\frac{N+i}{N+i-1}\nonumber\\
&=  \lg\left(1+ \frac{1}{N+[i-\gamma]_+-1}\right)-\lg\left(1+\frac{1}{N+i-1}\right) \nonumber\\
&>^{*2}  (\lg \text{e})\left[\frac{1}{N+[i-\gamma]_+-1}-\frac{1}{2(N+[i-\gamma]_+-1)^2}-\frac{1}{N+i-1} \right]\nonumber\\
&\geq  (\lg \text{e})\left[\frac{\gamma}{(N+[i-\gamma]_+-1)(N+i-1)}-\frac{1}{2(N+[i-\gamma]_+-1)^2}\right]\nonumber\\
&>  (\lg \text{e})\left[\frac{\gamma}{4N^2}-\frac{1}{2N^2}\right]\nonumber\\
& =\frac{(\gamma-2)\lg \text{e}}{4N^2}.
\label{eq6-21}
\end{align}
Hence, in the case of $Q(\alpha_i)\leq Q^*_{\gamma}(\alpha_i)$ for all $i$, the average codeword length $L$ 
satisfies \eqref{eqL4-3} from Lemma \ref{lm-2} and \eqref{eq6-19}.

\vspace{0.5cm}
\subsection{Proof of Theorem \ref{th-6}}\label{App-H}

In the case of  $N_s/N=p(s)$,  
the following relation holds from \eqref{eq6-0}.
\begin{align}
2^{k_s-1}< \frac{N}{N_s}=\frac{1}{p(s)} \leq 2^{k_s}. \label{eq-B1}
\end{align}
Then, we show that for $\mathcal{X}_s=\{\s_1, \s_2, \cdots, \s_{N_s}\}\subset \mathcal{X}=\{\alpha_1, \alpha_2, \cdots,
\alpha_N\}$, 
each $Q(\s_j)$ satisfies $Q(\s_j)\leq Q^*_{\gamma}(\alpha_{i'})$ for a certain $i', 1\leq i'\leq N$
when each $Q(\alpha_i)$ satisfies $Q(\alpha_i)\leq Q^*_{\gamma}(\alpha_i)=Q^*(\alpha_i)+\Theta(1/N^2)$
where $\Theta(1/N^2)>0$.
For $\lg$ function, we use the relation 
\begin{align}
 \frac{\lg \text{e}}{u} -\lg \left(1+\frac{1}{u}\right)=\Theta\left(\frac{1}{u^2}\right),
\label{eq-B2}
\end{align}
which satisfies $\zeta\times\Theta(1/u^2)=\Theta(1/u^2)$ for any constant $\zeta>0$.

First we consider the case of $1\leq j \leq A_s$.
\begin{align}
Q(\s_j)&=p(s) \sum_{i=B_s+(j-1)2^{k_s-1}+1}^{B_s+j 2^{k_s-1}}Q(\alpha_i)\hspace{3cm} \nonumber\\
&\leq \frac{N_s}{N} \sum_{i=B_s+(j-1)2^{k_s-1}+1}^{B_s+j 2^{k_s-1}} Q^*_{\gamma}(\alpha_i)\nonumber\\
&= \frac{N_s}{N} \sum_{i=B_s+(j-1)2^{k_s-1}+1}^{B_s+j 2^{k_s-1}} \left[\lg \frac{N+i}{N+i-1} +\Theta\left(\frac{1}{N^2}\right) \right]\nonumber\\
&=\frac{N_s}{N} \left[\lg \frac{N+B_s+j 2^{k_s-1}}{N+B_s+(j-1)2^{k_s-1}} +2^{k_s-1}\Theta\left(\frac{1}{N^2}\right)\right] \nonumber\\
&=\frac{N_s}{N} \left[\lg \left(1+\frac{2^{k_s-1}}{N+B_s+(j-1)2^{k_s-1}}\right)+2^{k_s-1}\Theta\left(\frac{1}{N^2}\right)\right] \nonumber\\
&<^{*3}\frac{N_s}{N} \left[(\lg e) \frac{2^{k_s-1}}{N+B_s+(j-1)2^{k_s-1}}+2^{k_s-1}\Theta\left(\frac{1}{N^2}\right)\right] \nonumber\\
&<^{*4}(\lg e) \frac{1}{N+B_s+(j-1)2^{k_s-1}} + \Theta\left(\frac{1}{N^2}\right) \nonumber\\
&=^{*3}\lg \left(1+\frac{1}{N+B_s+(j-1)2^{k_s-1}}\right)+ \Theta\left(\frac{1}{N^2}\right)\nonumber\\
&=\lg \frac{N+B_s+(j-1)2^{k_s-1}+1}{N+B_s+(j-1)2^{k_s-1}} +\Theta\left(\frac{1}{N^2}\right)  \nonumber\\
&= Q^*_{\gamma}(\alpha_{i'}),\label{eq-B3}
\end{align}
where $<^{*3}$ and $=^{*3}$ hold from \eqref{eq-B2}\footnote{
Note that for constants $\zeta_1, \zeta_2>0$, we have $\zeta_1\Theta(1/N^2)+\zeta_2\Theta(1/N^2)=\Theta(1/N^2)$. However since $\zeta_1\Theta(1/N^2)-\zeta_2\Theta(1/N^2)$ is not always positive, an inequality is used at $<^{*3}$.},
$<^{*4}$ holds since we have $1/2\leq (N_s/N)2^{k_s-1}<1$ from \eqref{eq-B1}, and $i'$ is defined by $i'=B_s+(j-1)2^{k_s-1}+1$.

Next we derive  $Q(\s_{A_s+j})$ for $1\leq j \leq N_s-A_s-1$. 
By letting $M_s=\left\lfloor N/N_s\right\rfloor$, we have
\begin{align}
Q(\s_{A_s+j})
&=p(s) \sum_{i=2^{k_s}N_s-N+(j-1)2^{k_s}+1}^{2^{k_s}N_s-N+j 2^{k_s}}Q(\alpha_i) \nonumber\\
&\leq \frac{N_s}{N} \sum_{i=2^{k_s}N_s-N+(j-1)2^{k_s}+1}^{2^{k_s}N_s-N+j 2^{k_s}} 
Q^*_{\gamma}(i)\nonumber\\
&= \frac{N_s}{N} \sum_{i=2^{k_s}N_s-N+(j-1)2^{k_s}+1}^{2^{k_s}N_s-N+j 2^{k_s}} \left[\lg \frac{N+i}{N+i-1} +\Theta\left(\frac{1}{N^2}\right) \right]\nonumber\\
&=\frac{N_s}{N} \left[\lg \frac{2^{k_s}N_s+j 2^{k_s}}{2^{k_s}N_s+(j-1)2^{k_s}} +2^{k_s}\Theta\left(\frac{1}{N^2}\right)\right] \nonumber\\
&=\frac{N_s}{N} \left[\lg \left(1+\frac{1}{N_s+(j-1)}\right)+2^{k_s}\Theta\left(\frac{1}{N^2}\right)\right] \nonumber\\
&<^{*3}\frac{N_s}{N} \left[(\lg e) \frac{1}{N_s+(j-1)}+2^{k_s}\Theta\left(\frac{1}{N^2}\right)\right] \nonumber\\
&=^{*5}(\lg e) \frac{1}{N+(j-1)\frac{N}{N_s}} + \Theta\left(\frac{1}{N^2}\right) \nonumber\\
&\leq^{*6}(\lg e) \frac{1}{N+(j-1)M_s} + \Theta\left(\frac{1}{N^2}\right) \nonumber\\
&=^{*3}  \lg \left(1+\frac{1}{N+(j-1)M_s}\right)+ \Theta\left(\frac{1}{N^2}\right)\nonumber\\
&=\lg \frac{N+(j-1)M_s+1}{N+(j-1)M_s} +\Theta\left(\frac{1}{N^2}\right) \hspace*{1.5cm} \nonumber\\
&= Q^*_{\gamma}(i'), \label{eq-B4}
\end{align}
where $=^{*5}$ holds since we have  $1\leq (N_s/N)2^{k_s}<2$ from \eqref{eq-B1},
$\leq^{*6}$ holds from $M_s\leq N_s/N$, and $i'$ is defined by $i'=(j-1)M_s+1$.

Finally we derive $Q(\s_j)$ for $j=N_s$.
\begin{align}
Q(\s_j)
&= p(s)\left[\sum_{i=1}^{B_s} Q(\alpha_i)+\sum_{i=N-D_s+1}^N Q(\alpha_i)\right]\nonumber\\
&\leq \frac{N_s}{N}\left[\sum_{i=1}^{B_s}\lg \frac{N+i}{N+i-1}+\sum_{i=N-D_s+1}^N \lg \frac{N+i}{N+i-1}
+(B_s+D_s)\Theta\left(\frac{1}{N^2}\right)\right]\nonumber\\
&=\frac{N_s}{N}\left[\lg \frac{N+B_s}{N}+\lg \frac{2N}{2N-D_s}  +(B_s+D_s)\Theta\left(\frac{1}{N^2}\right)\right]\nonumber\\
&=\frac{N_s}{N}\left[\lg \left(1+ \frac{2B_s+D_s}{2N-D_s}\right)+(B_s+D_s)\Theta\left(\frac{1}{N^2}\right)\right]\nonumber\\
&=^{*7}\frac{N_s}{N}\left[\lg \left(1+ \frac{1}{N_s+C_s}\right)+(B_s+D_s)\Theta\left(\frac{1}{N^2}\right)\right]\nonumber\\
&<^{*3} \frac{N_s}{N}\left[(\lg e) \frac{1}{N_s+C_s}+(B_s+D_s)\Theta\left(\frac{1}{N^2}\right)\right]\nonumber\\
&=^{*8}(\lg e)\frac{1}{N+C_s\frac{N}{N_s}}+\Theta\left(\frac{1}{N^2}\right)\nonumber\\
&\leq^{*6} (\lg e)\frac{1}{N+C_sM_s}+\Theta\left(\frac{1}{N^2}\right)\nonumber\\
&=^{*3}  \lg \left(1+\frac{1}{N+C_sM_s}\right)+\Theta\left(\frac{1}{N^2}\right)\nonumber\\
&=  \lg \frac{N+C_sM_s+1}{N+C_sM_s}+\Theta\left(\frac{1}{N^2}\right)\nonumber\\
&=Q^*_{\gamma}(\alpha_{i'}), \label{eq-B5}
\end{align}
where $=^{*7}$ holds from \eqref{eq6-3} and relation $2N-2^{k_s}N_s-D_s=C_s2^{k_s}$ shown 
 in Fig.~\ref{fig6},
$=^{*8}$ holds since we have $0< (N_s/N)(B_s+D_s)<(N_s/N)2^{k_s}<2$ from \eqref{eq6-3} and \eqref{eq-B1},
and $i'$ is defined by $i'=C_sM_s+1$.

We can conclude from \eqref{eq-B3}--\eqref{eq-B5} that
every $Q(\s_j)$ satisfies $Q(\s_j)\leq Q^*_{\gamma}(\alpha_{i'})$  if $Q(\alpha_i)\leq Q^*_{\gamma}(\alpha_i)$ for all $i$.
If $p(s)=p(s')$ for $s\ne s'$, it may occur that $Q(\s_j)=Q(\s'_{j'})$.
However,  if $Q(\s_j)\leq Q^*_{\gamma}(\alpha_{i'})$, $Q(\s_j)$ also satisfies $Q(\s_j)\leq Q^*_{\gamma}(\alpha_i)$
for $1\leq i \leq i'$. Hence, even in the case that several $Q(\s_j)$ take the same value,
$\{Q(\s_j), 1\leq j \leq N_s, s\in\mathcal{S}\}$ can be fitted to $\{Q^*_{\gamma}(\alpha_i), 1\leq i\leq N\}$.

From all of the above cases, we can conclude that the stationary probability distribution $\{Q(\alpha_i)\}$ of the sAEDS shown in Fig.~\ref{fig6} satisfies $Q(\alpha_i)\leq Q^*_{\gamma}(\alpha_i)$ for all $\alpha_i\in\mathcal{X}$.
Hence, from Lemma \ref{lm-4}, the sAEDS in Fig.~\ref{fig6} satisfies \eqref{eq6-23}.

\end{document}